\documentclass[12 pt]{article}
\usepackage{indentfirst,mathrsfs}
\usepackage{amsfonts}
\usepackage{amsmath,amsthm,amssymb}
\usepackage[colorlinks,linkcolor=black,
            pdftitle={title},
              pdfauthor={author},
              pdfkeywords={}]{hyperref}

\newtheorem{theorem}{Theorem}
\newtheorem{corollary}{Corollary}
\newtheorem{definition}{Definition}
\newtheorem{example}{Example}
\newtheorem{remark}{Remark}

\newtheorem{proposition}{Proposition}

\newtheorem{lemma}{Lemma}
\newcommand{\F}{\mathbb{F}}
\newcommand{\Z}{\mathbb{Z}}
\newcommand{\C}{\mathcal{C}}

\begin{document}
\begin{center}
\Large{Binary Linear Codes From Vectorial Boolean Functions and Their Weight Distribution}
\end{center}

\begin{center}
Deng Tang\footnotemark[1]~~~~~Claude Carlet\footnotemark[2]~~~~~Zhengchun Zhou\footnotemark[1]$^{,*}$
{\let\thefootnote\relax\footnotetext{$*$ Corresponding author.}}
\footnotetext[1]{School of Mathematics, Southwest Jiaotong University, Chengdu, 610031, China.
Email: dtang@foxmail.com (D. Tang), zzc@home.swjtu.edu.cn (Z. Zhou)}

\footnotetext[2]{LAGA, Department of Mathematics, University of
Paris 8 (and Paris 13 and CNRS), Saint--Denis cedex 02, France.
E-mail: claude.carlet@univ-paris8.fr}
\end{center}

\noindent\textbf{Abstract}
Binary linear codes with good parameters have important applications in secret sharing schemes,
authentication codes, association schemes, and consumer electronics and communications.
In this paper, we construct several classes of binary linear codes from vectorial Boolean functions and determine
their parameters, by further studying a generic construction developed by Ding \emph{et al.} recently.
First, by employing perfect nonlinear functions and almost bent functions, we obtain several classes of
six-weight linear codes which contains the all-one codeword.
Second, we investigate a subcode of any linear code mentioned above and consider its parameters.
When the vectorial Boolean function is a perfect nonlinear function or a Gold function in odd dimension,
we can completely determine the weight distribution of this subcode.
Besides, our linear codes have larger dimensions than the ones by Ding \emph{et al.}'s generic construction.

\vspace{5mm}

\section{Introduction}

Boolean functions are the building blocks of symmetric-key cryptography and coding theory.
Symmetric-key cryptography includes block ciphers and stream ciphers. 
Boolean functions are used to create substitution boxes (S-boxes) with
good cryptographic properties (low differential uniformity, high algebraic degree, large nonlinearity and so on)
in block ciphers, and utilized as nonlinear filters and combiners in stream ciphers when they satisfy the main cryptographic
criteria simultaneously.
Error-correcting codes have long been known to have many applications in CD players, high speed communications, cellular phones,
and (massive) data storage devices.
They have been widely studied by many researchers in the past four decades and
a lot of progresses on coding theory have been made.
Particularly, Boolean functions can be used to construct binary linear and nonlinear codes, and
 the two related classes of binary codes which are the
Reed-Muller \cite{ReedIT1954,MullerIT1954} codes and Kerdock codes  \cite{Kerdock72a,CarletKerdock88,CarletKerdock91} are well-known.

Binary linear codes with good parameters have wide applications in secret sharing schemes \cite{AndersonDHKDCC98,CarletDYIT05,YuanDIT06,DingDIT15},
authentication codes \cite{DingW05,DingHKWIT07}, and association schemes \cite{DelsarteL98},
in addition to their applications in consumer electronics and communications.
In the past two decades, the design of binary linear codes derived from (vectorial) Boolean functions
has been a research topic of increasing importance.
Many classes of binary codes with good parameters
have been obtained, for instance see \cite{DingIT2015,DingDM16} and the references therein.
Generally speaking, there are two generic constructions of binary linear from (vectorial) Boolean functions.
The first one is based on highly nonlinear vectorial Boolean functions
such as perfect nonlinear (PN) functions and almost bent functions.
For an $(m, s)$-function, this construction can generate a linear code with length $2^m$ and dimension at most $m+s+1$.
The second construction is based on the support set of some highly nonlinear Boolean functions such as bent and semi-bent functions.
For an $m$-variable Boolean function, such construction can provide
a linear code such that the length of this code is equal to the size of the support of this function and the dimension
of this code is at most $m+1$ \cite{DingIT2015,DingDM16}.
In the present paper, we extent the second construction to vectorial Boolean functions.
In general, for an $(m, s)$-function $F$ and an arbitrary component function $f$, we can construct
a linear code such that the length of this code is equal to the size of the support of $f$, the dimension of this code is at most $m+s$,
and the minimum Hamming distance can be expressed by means of the nonlinearity of $F$.
When $F$ is an arbitrary PN function or AB function, we can determine the length, dimension, and weight distribution
of the linear codes generated by our construction. All of these codes are six-weight linear codes and contain the all-one codeword;
When $m$ is small, some examples of our linear codes are optimal or at least have the same parameters as the best known codes \cite{Grassl:codetables}
(see some examples in Remark \ref{R:PN}).
Further, we define a subcode of any linear code given above and consider its weight distribution.
When $F$ is a PN function or Gold function in odd dimension, we can completely determine the weight distribution of such subcode.

The remainder of this paper is organized as follows. In Section \ref{sec:pre}, the notation
and the necessary preliminaries required for the subsequent sections are reviewed.
In Section \ref{sec:constcode}, we present the construction of linear codes from vectorial Boolean
functions and provide the parameters of those codes.
In Section \ref{sec:subcode}, we focus on calculating the weight distributions of some subcodes of
the codes given in Section \ref{sec:constcode}.
Finally, Section \ref{sec:conc} concludes the paper.

\section{Preliminaries}\label{sec:pre}
For any positive integer $m$, we denote by $\F_2^m$ the vector space of $m$-tuples over the
finite field $\F_2=\{0,1\}$, and by $\F_{2^m}$ the finite field of order $2^m$.
For simplicity, we denote by $\F_2^{m*}$ the set $\F_2^{m}\setminus\{(0,0,\cdots,0)\}$, and
 $\F_{2^m}^*$ denotes  the set $\F_{2^m}\setminus\{0\}$.
We use $+$ (resp. $\sum$) to denote the addition (resp. a multiple sum) in $\mathbb{Z}$ or in the finite field $\F_{2^m}$, and
$\oplus$ (resp. $\bigoplus$) to denote the addition (resp. a multiple sum) in $\F_2$.
For simplicity, when there will be no ambiguity, we shall allow us to use $+$ instead of $\oplus$.
The vector space $\F_2^m$ is isomorphic to the finite field $\F_{2^m}$ through the choice of some basis of $\F_{2^m}$ over $\F_2$.
Indeed, let $(\lambda_1, \lambda_2,\cdots,\lambda_m)$ be a basis of $\F_{2^m}$ over $\F_2$,
then every vector $x=(x_1,\cdots,x_m)$ of $\F_2^m$ can be identified with the element $x_1\lambda_1+x_2\lambda_2+\cdots+x_m\lambda_m\in\F_{2^m}$.
The finite field $\F_{2^m}$ can then be viewed as an $m$-dimensional vector space over $\F_2$; each of its elements can be identified with a binary vector of length $m$,
the element $0\in\F_{2^m}$ is identified with the all-zero vector.
We shall use $x$ to denote indifferently elements in $\F_{2^m}$ and in $\F_2^m$.

A Boolean function of $m$ variables is a function from $\F_2^m$ into $\F_2$.
We denote by $\mathcal{B}_m$ the set of Boolean functions of $m$ variables.
Any Boolean function $f\in\mathcal{B}_m$ can be expressed by its truth table, i.e.,
\begin{eqnarray*}
f=\big[f(0,  \cdots, 0,0), f(0, \cdots ,0, 1), \cdots, f(1, \cdots,1,0), f(1, \cdots,1,  1)\big].
\end{eqnarray*}
We say that a Boolean function $f\in\mathcal{B}_m$ is balanced if
its truth table contains an equal number of ones and zeros, that is,
if its Hamming weight equals $2^{m-1}$, where the Hamming weight of $f$, denoted by $n_f$, is
defined as the size of the support of $f$ in which the support of $f$
is defined as $D_f=\{x\in\F_2^m | f(x)\not=0\}$.
Given two Boolean functions $f$ and $g$ in $m$ variables, the
Hamming distance between $f$ and $g$ is defined as $d_H(f,g)=|\{x\in\F_2^m \,|\, f(x)\ne g(x) \}|$.
Any Boolean function $f$ of $m$ variables can also be expressed
in terms of a polynomial in $\F_2[x_1,\cdots,x_m]/(x_1^2\oplus x_1,\cdots,x_m^2\oplus x_m)$:
\begin{eqnarray*}\label{D:ANF}
f(x_1,\cdots,x_m)&=&\bigoplus_{u\in\F_2^m}a_u\Big{(}\prod_{j=1}^n x_j^{u_j}\Big{)}=\bigoplus_{u\in\F_2^m}a_u x^u,\\
\end{eqnarray*}
where $a_u\in\F_2$. This representation is called the \emph{algebraic normal form} (ANF).
The \textit{algebraic degree}, denoted by ${\rm deg}(f)$, is the maximal value of ${\rm wt}(u)$ such that $a_u\ne 0$.
A Boolean function is an \emph{affine function} if its algebraic degree is at most 1.
The set of all affine functions of $m$ variables is denoted by $A_m$.
Recall that $\F_{2^m}$  is isomorphic as a $\F_2$-vector space to $\F_2^m$ through the
choice of some basis of $\F_{2^m}$ over $\F_2$. The Boolean
functions over $\F_{2^m}$ can also be uniquely expressed by a
\textit{univariate polynomial}
\begin{equation*}\label{Ploynoimal}
f(x) =\sum_{i=0}^{2^m-1} a_i x^i,
\end{equation*}
where $a_0,a_{2^m-1}\in\F_2$, $a_{i}\in\F_{2^m}$ for $1\le i<2^m-1$
such that $a_i=a_{2i\, [{\rm mod}\; 2^m-1]}$,  and the addition is
modulo $2$. The algebraic degree $\textrm{deg}(f)$ under this representation
is equal to $\max\{\mathrm{wt}(\overline{i}) : a_i\ne 0, 0\le i<2^m\}$,
where $\overline{i}$ is the binary expansion of $i$ (see e.g. \cite{Carlet10}).
The\textit{ nonlinearity} $nl(f)$ of a Boolean function $f\in\mathcal{B}_m$ is defined as the minimum Hamming distance between $f$ and all the affine functions:
\begin{equation*}
nl(f) = \min_{g\in A_m}(d_H(f,g)).
\end{equation*}
In other words, the nonlinearity of a Boolean function $f$ in $m$ variables equals the minimum Hamming
distance between the binary vector of length $2^m$ listing the values of the function
to the Reed-Muller code $\textrm{RM}(1,m)$ of length $2^m$.
The maximal nonlinearity of all Boolean functions in $m$ variables equals by definition
the covering radius of $\textrm{RM}(1,m)$ \cite{CHLL97}.
Let $x=(x_1, x_2, \cdots , x_m)$ and $a=(a_1, a_2, \cdots , a_m)$
both belong to $\F_2^m$ and let $a\cdot x$ be any inner product; for instance the usual inner product
$a\cdot x=a_1x_1 \oplus a_2x_2 \oplus \cdots \oplus a_mx_m$,
then the \emph{Walsh transform} of a Boolean function $f\in\mathcal{B}_m$ at point $a$ is defined by
\begin{equation*}
\widehat{f}(a)=\sum_{x\in\F_2^m}(-1)^{f(x)+a\cdot x}.
\end{equation*}
Note that changing the inner product changes the order of the values of the Walsh transform
but not the multiset of these values which is called the \emph{Walsh spectrum}.
Over $\F_{2^m}$, the Walsh transform of the Boolean function $f$ at $\alpha\in\F_{2^m}$ can be defined by
\begin{equation*}
\widehat{f}(\alpha)=\sum_{x\in\F_{2^m}}(-1)^{f(x)+Tr_1^m(\alpha x)},
\end{equation*}
where $Tr_t^m(x)=\sum_{i=0}^{m/t-1}x^{2^{ti}}$ is the trace function from $\F_{2^m}$ to $\F_{2^t}$
in which $t$ is a positive divisor of $m$.
The well-known Parseval relation \cite{MS1977} states that: for any $m$-variable Boolean function,
we have $\sum_{u\in\F_2^m}{\widehat{f}}^2(u)=2^{2m}$.
Parseval's relation implies that, for a Boolean function of $m$ variables,
the mean of square of Walsh spectrum equals $2^m$.
Then the maximum of the square of Walsh spectrum is greater than or equal to $2^m$
and therefore $\max_{u\in\F_2^m}|\widehat{f}(u)|\geq 2^{\frac{m}{2}}$.
This implies that the nonlinearity of Boolean functions in $m$ variables is upper-bounded by $2^{m-1}-2^{\frac{m}{2}-1}$,
which is tight for even $m$.
\begin{definition}(\cite{Rothaus})\label{D:bent}
Let $m$ be an even integer and $f$ be a Boolean function of $m$ variables.
If $nl(f)=2^{m-1}-2^{\frac{m}{2}-1}$, then we say that $f$ is \emph{bent}.
\end{definition}
\noindent For odd number of variables $m$, the maximum nonlinearity of an $m$-variable Boolean
functions for $m$ odd is $2^{m-1}-2^\frac{m-1}{2}$ when $m=1,3,5,7$ and
the question for the maximum nonlinearity of functions in odd $m\ge 9$ variables is still completely open.
\begin{definition}\label{D:semi-bent}
Let $m$ be an odd integer and $f$ be a Boolean function of $m$ variables.
If the set formed by the Walsh spectrum of $f$ equals $\{0,\pm 2^{(m+1)/2}\}$, then we say that $f$ is \emph{semi-bent} (or \emph{near-bent} ).
\end{definition}

\noindent The nonlinearity of a Boolean function $f\in\mathcal{B}_m$ can be calculated as
\begin{eqnarray}\label{nlwalsh}
nl(f)&=& 2^{m-1}-\frac{1}{2}\max_{a\in\F_2^m}|W_f(a)|\nonumber\\
&=&2^{m-1}-\frac{1}{2}\max_{\omega\in\F_{2^m}}|W_f(\omega)|\nonumber.
\end{eqnarray}

Given two integers $m$ and $s$, a mapping from $\F_{2^m}$ to $\F_{2^s}$ which is often called
an $(m,s)$-function or a vectorial Boolean function if the values
$m$ and $s$ are omitted, can be viewed  (and vice versa) as a function $G$ from the
vectorial space $\F_2^m$ to the vectorial space  $\F_2^s$.
Particularly, $G$ is called a Boolean function when $s=1$.
Let $G$ be an $(m,s)$-function, the Boolean functions
$g_1(x),\cdots,g_s(x)$ of $m$ variables defined by $G(x) =(g_1(x),\cdots, g_s(x))$
are called the \emph{coordinate} functions of $G$. Further, the Boolean
functions, which are the linear combinations, with non all-zero
coefficients of the coordinate functions of $G$, are called
\emph{component} functions of $G$. The component functions of $G$ can be
expressed as $a\cdot G$ where $a\in\F_2^{s*}$. If we identify every
element of $\F_2^s$ with an element of finite field $\F_{2^s}$, then
the component functions $g_\alpha$ of $G$ can be expressed as $Tr_1^s(\alpha
G)$, where $\alpha\in\F_{2^s}^*$.
For any $(a,b)\in\F_2^{s*}\times\F_2^m$, the Walsh transform of $G$ at $(a,b)$ is defined as
\begin{eqnarray*}
\widehat{G}(a,b)=\sum_{x\in\F_2^m}(-1)^{a\cdot G(x)+b\cdot x}.
\end{eqnarray*}
If $(\alpha,\beta)\in\F_{2^s}^*\times\F_{2^m}$, the Walsh transform of $G$ at $(\alpha,\beta)$ is defined as
\begin{eqnarray*}
\widehat{G}(\alpha,\beta)=\sum_{x\in\F_2^m}(-1)^{Tr_1^s(\alpha G(x))+Tr_1^m(\beta x)}.
\end{eqnarray*}
We call extended Walsh spectrum of $G$ (and we shall denote by $EW_{G}$)
the multi-set of the absolute values of all the Walsh transform of $G$.
The nonlinearity $nl(G)$ of an $(m,s)$-function $G$ is the minimum Hamming
distance between all the component functions of $G$ and all affine functions in $m$ variables.
According to the definition of Walsh transform, we have
\begin{eqnarray}\label{D:nlF}
nl(G)&=&2^{m-1}-\frac 12 \max_{(a,b)\in\F_2^{s*}\times\F_2^m} |\widehat{G}(a,b)|\nonumber\\
&=&2^{m-1}-\frac 12 \max_{(\alpha,\beta)\in\F_{2^s}^*\times\F_{2^m}} |\widehat{G}(\alpha,\beta)|\nonumber.
\end{eqnarray}
The nonlinearity $nl(G)$ is upper-bounded by $2^{m-1}-2^{\frac{m-1}{2}}$ when $m=s$.
This upper bound is tight for odd $m=s$.
For even $n=m$, the best known value of the nonlinearity of $(n,m)$-functions is $2^{n-1}-2^{\frac{n}{2}}$.

\begin{definition}\label{D:AB}
Let $m$ be an odd integer and $G$ be an $(m,m)$-function.
If $nl(G)=2^{m-1}-2^{\frac{m-1}{2}}$, then $G$ is called almost bent (AB).
\end{definition}
\noindent It is well-known that the extended Walsh spectrum values of
an almost bent $(m,m)$-functions $G$ are $0$ and $ 2^{(m+1)/2}$ and thus an $(m,m)$-function
is almost bent if and only if all of its component functions are semi-bent.
\begin{definition}\label{D:PN}
For two integers $m$ and $s$,
an $(m,s)$-function is called bent vectorial if it nonlinearity is equal to $2^{m-1}-2^{m/2-1}$
\end{definition}
\noindent
Clearly, an $(m,s)$-function is bent vectorial if and only if all of its component
functions are bent. The bent vectorial functions exist only for even $m$ and $s\leq m/2$ \cite{Nyberg1991PN}. They are characterized by the fact that all their derivatives $D_aF(x)=F(x)+F(x+a)$, $a\in \F_{2^m}^*$, are balanced (i.e. take each value of $\F_{2^s}$ the same number of times $2^{m-s}$) and are then also called perfect nonlinear (PN).

\section{Linear codes from vectorial Boolean functions}\label{sec:constcode}

In this section, we will give a method for obtaining linear codes from vectorial Boolean
functions and get the parameters of those codes.

Let us first recall some basic definitions related to linear codes.
Let $p$ be a prime, $m$ a positive integer, $r$ a positive divisor of $m$
and $q=p^r$. An $[n, k, d]_q$ linear code $\mathcal{C}$ over $\F_q$
is a $k$-dimensional subspace of $\F_q^n$ with minimum Hamming distance $d$. Recall that
$d=\min_{a,b\in\mathcal{C}}d_H(a,b)$ where $d_H$ denotes the Hamming distance
between vectors (called codewords) $a=(a_1,a_2,\cdots,a_n)\in\mathcal{C}$ and $b=(b_1,b_2,\cdots,b_n)\in\mathcal{C}$,
i.e., $d_H(a,b)=|\{1\leq i\leq n : a_i\not=b_i\}|$.
For a given codeword $a=(a_1,a_2,\cdots,a_n)\in\mathcal{C}$, the Hamming weight ${\rm wt}(a)$
is defined as the number of nonzero coordinates.
A generator matrix $G$ of a linear $[n, k, d]$ code $\mathcal{C}$ is a
$k\times n$ matrix whose rows form a basis of $\mathcal{C}$.
The dual code $\mathcal{C}^\perp$ is the orthogonal subspace under the usual inner product in $\F_q^n$.
Usually, if the context is clear we omit the subscript $q$ by convention in the sequel.
Highly nonlinear Boolean and vectorial Boolean functions (or more generally functions valued in $\F_q$) have important applications
in cryptography and coding theory. In coding theory, such functions have been used to construct linear codes
with good parameters, for instance in papers \cite{DingIT2015,DingDM16,MesnagercodeBF15,TangQH16ICL,Xu2016CCDS,LCJAAECC2016,ZhouLFH16DCC,HengYL16DM}.

Let $m,s$ be two integers and $F$ be a vectorial Boolean function from $\F_{2^m}$ to $\F_{2^s}$.
For any $\lambda\in\F_{2^s}^*$, we denote by $f_\lambda$ the Boolean function $Tr_1^s(\lambda F)$ which is
a component function of $F$.
Recall from Section \ref{sec:pre} that $n_{f_\lambda}$ denotes the
Hamming weight of $f_\lambda$ and $D_{f_\lambda}$ denotes the support of $f_\lambda$.
Let $D_{f_\lambda}=\{d_1,d_2,\cdots,d_{n_{f_\lambda}}\}$, we
define a linear code of length $n_{f_\lambda}$ over $\F_2$ as follows:
\begin{eqnarray}\label{E:codeVBFall}
\mathcal{C}_{D_{f_\lambda}}=\{c_{x,y}: x\in \F_{2^m},y\in \F_{2^s}\},
\end{eqnarray}
where $c_{x,y}=\big(Tr_1^m(xd_1)+Tr_1^s(yF(d_1)),\cdots, Tr_1^m(xd_n)+Tr_1^s(yF(d_{n_{f_\lambda}}))\big)$.

We can easily see that the code $\mathcal{C}_{D_{f_\lambda}}$ is linear.
For any $(m,s)$-function $G$, its graph is defined as the set $\{(x,y)\in\F_2^m\times\F_2^s:y=G(x)\}$.
Then we can see that the codewords of the linear code $\mathcal{C}_{D_{f_\lambda}}$
are the evaluations of all linear functions at those elements of the graph of $F$ whose abscissa
belong to the support of some fixed component function of $F$.
Hence the codewords are the restrictions of all codewords of the extended simplex code of
length $2^{m}$ to such elements of the graph of $F$.

In the rest of this section, we give the parameters of the linear code $\mathcal{C}_{D_{f_\lambda}}$,
which are heavily relied on the extended Walsh spectrum of $F$.


\begin{proposition}\label{T:codeVBFall}
Let $F$ be an $(m,s)$-function. For any $\lambda\in\F_{2^s}^*$, let $f_\lambda=Tr_1^s(\lambda F)$ and let $n_{f_\lambda}$ be the Hamming weight of $f_\lambda$, equal to the size of the support $D_{f_\lambda}$ of $f_\lambda$. If $2^m-2nl(F)<n_{f_\lambda}$,
then the linear code $\mathcal{C}_{D_{f_\lambda}}$ defined by \eqref{E:codeVBFall} has length $n_{f_\lambda}$,
dimension $m+s$ and minimum Hamming weight no less than $nl(F)-\frac{2^m-n_{f_\lambda}}{2}$.
\end{proposition}
\begin{proof}
By \eqref{E:codeVBFall}, it is clear that every codeword in $\mathcal{C}_{D_{f_\lambda}}$ has length $n_{f_\lambda}$.
We now prove that the linear code $\mathcal{C}_{D_{f_\lambda}}$ has
dimension $m+s$ if $2^m-2nl(F)<n_{f_\lambda}$.
For doing this, we only need to prove that for any two distinct pairs
$(x_1,y_1),(x_2,y_2)\in\F_{2^m}\times\F_{2^s}$ the Hamming distance between codewords
$c_{x_1,y_1}$ and $c_{x_2,y_2}$ is not equal to zero, i.e., $d_H(c_{x_1,y_1},c_{x_2,y_2})\not=0$,
where $c_{x,y}$ is defined by \eqref{E:codeVBFall}.
Define
\begin{eqnarray*}
A&=&\sum_{d\in D_{f_\lambda}}(-1)^{\big(Tr_1^m(x_1d)+Tr_1^s(y_1F(d))\big)+\big(Tr_1^m(x_2d)+Tr_1^s(y_2F(d))\big)}\\
&=&\sum_{d\in D_{f_\lambda}}(-1)^{Tr_1^m((x_1+x_2)d)+Tr_1^s((y_1+y_2)F(d))}.
\end{eqnarray*}
Let us use $t(d)$ to denote $Tr_1^m((x_1+x_2)d)+Tr_1^s((y_1+y_2)F(d))$, we have
\begin{eqnarray*}
\left\{\begin{array}{lll}
|\{d\in D_{f_\lambda}\,: t(d)=0\}|-|\{d\in D_{f_\lambda}\,: t(d)=1\}|&=&A\\
|\{d\in D_{f_\lambda}\,: t(d)=0\}|+|\{d\in D_{f_\lambda}\,: t(d)=1\}|&=&n_{f_\lambda}\\
|\{d\in D_{f_\lambda}\,: t(d)=1\}|&=&d_H(c_{x_1,y_1},c_{x_2,y_2})
\end{array} \right..
\end{eqnarray*}
Thus we get
\begin{eqnarray}\label{E:pairsdH}
d_H(c_{x_1,y_1},c_{x_2,y_2})&=&\frac{1}{2}\big(n_{f_\lambda}-A\big).
\end{eqnarray}

We distinguish the following three cases to calculate the values of $d_H(c_{x_1,y_1},c_{x_2,y_2})$.

\noindent \textbf{Case 1}. $y_1+y_2=0$.

Obviously, in this case  $x_1+x_2\not=0$
and $y_1+y_2\not=\lambda$. Since $Tr_1^s(\lambda F(d))=1$ if $d\in D_{f_\lambda}$ and $Tr_1^s(\lambda F(d))=0$ otherwise, we have:
\begin{eqnarray*}
\left\{\begin{array}{lll}
\sum\limits_{d\in D_{f_\lambda}}(-1)^{Tr_1^m((x_1+x_2)d)}+\sum\limits_{d\in \F_{2^m}\setminus D_{f_\lambda}}(-1)^{Tr_1^m((x_1+x_2)d)}&=&0\\
\sum\limits_{d\in D_{f_\lambda}}(-1)^{Tr_1^m((x_1+x_2)d)+Tr_1^s(\lambda F(d))}+\sum\limits_{d\in \F_{2^m}\setminus D_{f_\lambda}}(-1)^{Tr_1^m((x_1+x_2)d)+Tr_1^s(\lambda F(d))}&=&\widehat{F}(\lambda,x_1+x_2)
\end{array} \right.,
\end{eqnarray*}
which is equivalent to
\begin{eqnarray*}
\left\{\begin{array}{lll}
A+\sum\limits_{d\in \F_{2^m}\setminus D_{f_\lambda}}(-1)^{Tr_1^m((x_1+x_2)d)}&=&0\\
-A+\sum\limits_{d\in \F_{2^m}\setminus D_{f_\lambda}}(-1)^{Tr_1^m((x_1+x_2)d)}&=&\widehat{F}(\lambda,x_1+x_2)
\end{array} \right..
\end{eqnarray*}
This implies that $A=-\frac 12\widehat{F}(\lambda,x_1+x_2)$. By \eqref{E:pairsdH}, we have
\begin{eqnarray}\label{E:A0}
d_H(c_{x_1,y_1},c_{x_2,y_2})&=&\frac{1}{4}\big(2n_{f_\lambda}+\widehat{F}(\lambda,x_1+x_2)\big).
\end{eqnarray}

\noindent \textbf{Case 2}. $y_1+y_2=\lambda$.

If $x_1+x_2\not=0$, similarly to Case 1, we have
\begin{eqnarray*}
\left\{\begin{array}{lll}
\sum\limits_{d\in D_{f_\lambda}}(-1)^{Tr_1^m((x_1+x_2)d)}+\sum\limits_{d\in \F_{2^m}\setminus D_{f_\lambda}}(-1)^{Tr_1^m((x_1+x_2)d)}&=&0\\
\sum\limits_{d\in D_{f_\lambda}}(-1)^{Tr_1^m((x_1+x_2)d)+Tr_1^s(\lambda F(d))}+\sum\limits_{d\in \F_{2^m}\setminus D_{f_\lambda}}(-1)^{Tr_1^m((x_1+x_2)d)+Tr_1^s(\lambda F(d))}&=&\widehat{F}(\lambda,x_1+x_2)
\end{array} \right.,
\end{eqnarray*}
which is equivalent to
\begin{eqnarray*}
\left\{\begin{array}{lll}
-A+\sum\limits_{d\in \F_{2^m}\setminus D_{f_\lambda}}(-1)^{Tr_1^m((x_1+x_2)d)}&=&0\\
A+\sum\limits_{d\in \F_{2^m}\setminus D_{f_\lambda}}(-1)^{Tr_1^m((x_1+x_2)d)}&=&\widehat{F}(\lambda,x_1+x_2)
\end{array} \right..
\end{eqnarray*}
and then $A=\frac 12 \widehat{F}(\lambda,x_1+x_2)$.
It follows from \eqref{E:pairsdH} that
\begin{eqnarray}\label{E:Alambda}
d_H(c_{x_1,y_1},c_{x_2,y_2})&=&\frac{1}{4}\big(2n_{f_\lambda}-\widehat{F}(\lambda,x_1+x_2)\big).
\end{eqnarray}
If $x_1+x_2=0$, we have $A=\sum_{d\in D_{f_\lambda}}(-1)^{1}=-n_{f_\lambda}$.
This implies that
\begin{eqnarray}\label{E:Alambda1}
d_H(c_{x_1,y_1},c_{x_2,y_2})&=&n_{f_\lambda} \mbox{~when~} x_1+x_2=0,
\end{eqnarray}
according to \eqref{E:pairsdH}.

\noindent \textbf{Case 3}. $y_1+y_2\in\F_{2^s}\setminus\{0,\lambda\}$.

Note that $\lambda+y_1+y_2\not=0$ in this case. We have
\begin{eqnarray*}
\left\{\begin{array}{lll}
\sum\limits_{d\in D_{f_\lambda}}(-1)^{{Tr_1^m((x_1+x_2)d)+Tr_1^s((y_1+y_2)F(d))}}+\sum\limits_{d\in \F_{2^m}\setminus D_{f_\lambda}}(-1)^{{Tr_1^m((x_1+x_2)d)+Tr_1^s((y_1+y_2)F(d))}}\\
~~~~~~~~~~~~~~~~~~~~~~~~~~~~~~~~~~~~~~~~~~~~~~~~=\widehat{F}(y_1+y_2,x_1+x_2)\\
\sum\limits_{d\in D_{f_\lambda}}(-1)^{{Tr_1^m((x_1+x_2)d)+Tr_1^s((y_1+y_2+\lambda)F(d))}}+\sum\limits_{d\in \F_{2^m}\setminus D_{f_\lambda}}(-1)^{{Tr_1^m((x_1+x_2)d)+Tr_1^s((y_1+y_2+\lambda)F(d))}}\\
~~~~~~~~~~~~~~~~~~~~~~~~~~~~~~~~~~~~~~~~~~~~~~~~~~~=\widehat{F}(y_1+y_2+\lambda,x_1+x_2)\\
\end{array} \right..
\end{eqnarray*}
Note that $Tr_1^s(\lambda F(d))=1$ if $d\in D_{f_\lambda}$ and $Tr_1^s(\lambda F(d))=0$ otherwise.
Then we have
\begin{eqnarray*}
\left\{\begin{array}{lll}
A+\sum\limits_{d\in \F_{2^m}\setminus D_{f_\lambda}}(-1)^{{Tr_1^m((x_1+x_2)d)+Tr_1^s((y_1+y_2)F(d))}}&=&\widehat{F}(y_1+y_2,x_1+x_2)\\
-A+\sum\limits_{d\in \F_{2^m}\setminus D_{f_\lambda}}(-1)^{{Tr_1^m((x_1+x_2)d)+Tr_1^s((y_1+y_2)F(d))}}&=&\widehat{F}(y_1+y_2+\lambda,x_1+x_2)\\
\end{array} \right..
\end{eqnarray*}
Thus, we have $A=\frac 12 \big(\widehat{F}(y_1+y_2,x_1+x_2)-\widehat{F}(y_1+y_2+\lambda,x_1+x_2)\big)$
and hence, by \eqref{E:pairsdH},
\begin{eqnarray}\label{E:Anot=0andlambda}
d_H(c_{x_1,y_1},c_{x_2,y_2})=\frac{1}{4}\big(2n_{f_\lambda}-\widehat{F}(y_1+y_2,x_1+x_2)+\widehat{F}(y_1+y_2+\lambda,x_1+x_2)\big).
\end{eqnarray}
Combining \eqref{E:A0},\eqref{E:Alambda} and \eqref{E:Alambda1}, we have
\begin{eqnarray}\label{E:pairsdHsetlambda0}
d_H(c_{x_1,y_1},c_{x_2,y_2})&\in&\Big\{\frac{1}{4}\big(2n_{f_\lambda}+\widehat{F}(\lambda,\alpha)\big),
\frac{1}{4}\big(2n_{f_\lambda}-\widehat{F}(\lambda,\alpha)\big),n_{f_\lambda}\Big\}
\end{eqnarray}
for $y_1+y_2\in\{0,\lambda\}$, and by \eqref{E:Anot=0andlambda} we have
\begin{eqnarray}\label{E:pairsdHsetnotlambda0}
d_H(c_{x_1,y_1},c_{x_2,y_2})&=&\frac{1}{4}\big(2n_{f_\lambda}-\widehat{F}(\gamma,\beta)+\widehat{F}(\gamma+\lambda,\beta)\big),
\end{eqnarray}
for $y_1+y_2\in\F_{2^s}\setminus\{0,\lambda\}$, in which
$\alpha\in\F_{2^m}^*$,$\gamma\in\F_{2^s}\setminus\{0,\lambda\}$,$\beta\in\F_{2^m}$.
So we have
\begin{eqnarray}\label{E:minimumdistance}
d_H(c_{x_1,y_1},c_{x_2,y_2})&\geq& \frac 12\Big(n_{f_\lambda}-\max_{(\mu,\nu)\in\F_{2^s}^*\times\F_{2^m}} |\widehat{F}(\mu,\nu)|\Big)\\
&>&0\nonumber.
\end{eqnarray}
The last inequality follows from the condition $2^m-2nl(F)<{n_{f_\lambda}}$ which means that
$\max_{(\mu,\nu)\in\F_{2^s}^*\times\F_{2^m}} |\widehat{F}(\mu,\nu)|<n_{f_\lambda}$ according to \eqref{D:nlF}.
Therefore, if $2^m-2nl(F)<{n_{f_\lambda}}$ we have
$c_{x_1,y_1}\not=c_{x_2,y_2}$ for any two distinct pairs
$(x_1,y_1),(x_2,y_2)\in\F_{2^m}\times\F_{2^s}$ and hence
$\mathcal{C}_{D_{f_\lambda}}$ has dimension $m+s$.
Furthermore, for any two distinct pairs
$(x_1,y_1),(x_2,y_2)\in\F_{2^m}\times\F_{2^s}$, we have
$d_H(c_{x_1,y_1},c_{x_2,y_2})\geq \frac 12\big(n_{f_\lambda}-\max_{(\mu,\nu)\in\F_{2^s}^*\times\F_{2^m}} |\widehat{F}(\mu,\nu)|\big)
=\frac 12 \big(n_{f_\lambda}-2^m+2nl(F)\big)=nl(F)-\frac 12\big(2^m-n_{f_\lambda}\big)$.
This implies that the minimum Hamming distance of the linear code
$\mathcal{C}_{D_{f_\lambda}}$ is no less than $nl(F)-\frac 12 \big(2^m-n_{f_\lambda}\big)$.

This completes the proof.
\end{proof}

\begin{theorem}\label{T:codeVBFallspectra}
Let $F$ be an $(m,s)$-function with $nl(F)>0$. For any integer $w\in EW_{F}$,
Rel. \eqref{E:codeVBFall} allows designing two linear codes
with parameters $[2^{m-1}-w/2,m+s,nl(F)-2^{m-2}-w/4]$ and $[2^{m-1}+w/2,m+s,nl(F)-2^{m-2}+w/4]$ respectively.
\end{theorem}
\begin{proof}
Suppose that the component functions of $F$ are $f_1,f_2,\cdots,f_{2^s-1}$.
If $w\in EW_{F}$, there exists a function $f_i$ $(1\leq i \leq 2^s-1)$ and
an element $a\in\F_{2^s}$ such that $|W_{f_i}(a)|=w$.
Clearly, there exist $s$ component functions $f_i,f_{j_1},f_{j_2},\cdots, f_{j_{s-1}}$
such that $lf_i+\sum_{i=1}^{s-1}l_if_{j_i}\not\in A_m$ for any $(l,l_1,\cdots,l_{s-1})\in\F_2^{s*}$,
since any function in $A_m$ (the set of $m$-variable Boolean functions with algebraic degree no more than $1$) has nonlinearity $0$ and
the $(m,s)$-function $F$ has nonlinearity greater than $0$.
Therefore, $l(f_i+Tr_1^m(ax)+c)+\sum_{i=1}^{s-1}l_if_{j_i}=lf_i+\sum_{i=1}^{s-1}l_if_{j_i}+Tr_1^m(ax)+c\not\in A_m$,
where $c\in\F_2$, for any $(l,l_1,\cdots,l_{s-1})\in\F_2^{s*}$.
Define an $(m,s)$-function $F'=(f_i+Tr_1^m(ax)+c,f_{j_1},f_{j_2},\cdots, f_{j_{s-1}})$.
It can be easily checked that $nl(F')=nl(F)$.

\noindent \textbf{Case A}. If $W_{f_i}(a)=w$, by taking $F=F'$ and $f_\lambda=f_i+Tr_1^m(ax)$ in \eqref{E:codeVBFall},
we have $n_{f_\lambda}=|D_{f_{\lambda}}|=2^m-|\{x\in\F_{2^m}:f_\lambda(x)=0\}|$
and $|\{x\in\F_{2^m}:f_\lambda(x)=0\}|-|D_{f_{\lambda}}|=\sum_{x\in\F_{2^m}}(-1)^{f_{\lambda}(x)}=\sum_{x\in\F_{2^m}}(-1)^{f_i(x)+Tr_1^m(ax)}=w$.
This implies that $n_{f_\lambda}=2^{m-1}-w/2$ and hence
we can get a $[2^{m-1}-w/2,m+s,nl(F)-2^{m-2}-w/4]$-code by Proposition \ref{T:codeVBFall}.
Similarly, by taking $F=F'$ and $f_\lambda=f_i+Tr_1^m(ax)+1$ in \eqref{E:codeVBFall}, we
can get a $[2^{m-1}+w/2,m+s,nl(F)-2^{m-2}+w/4]$-code by Proposition \ref{T:codeVBFall}.

\noindent \textbf{Case B}.  If $W_{f_i}(a)=-w$, we can get the same codes as Case A with
similar discussion.

This completes the proof.
\end{proof}


Let $F$ be a perfect nonlinear function from $\F_{2^m}$ to
$\F_{2^{m/2}}$, where $m$ is even. It is well-known that all
values in the extended Walsh spectrum of $F$ are equal to $2^{m/2}$. Thus, by
Theorem \ref{T:codeVBFallspectra}, we can immediately get the
following two corollaries.

\begin{corollary}\label{C:codeVBFallPN}
Let $F$ be a perfect nonlinear function from $\F_{2^m}$ to $\F_{2^{m/2}}$ where $m$ is even.
Then there exist two linear codes with parameters
$[2^{m-1}-2^{m/2-1}, 3m/2, 2^{m-2}-3\cdot 2^{m/2-2}]$
and $[2^{m-1}+2^{m/2}, 3m/2-1, 2^{m-2}-2^{m/2-2}]$ respectively.
\end{corollary}

We have mentioned that Nyberg proved in \cite{Nyberg1991PN} that an $(m,s)$-function is perfect nonlinear (or equivalently, bent) only if $m$ is even and $s\leq m/2$.
In fact, for any even integer $m\geq 4$, bent vectorial $(m,m/2)$-functions do exist.
For examples, we list below some primary constructions of bent vectorial $(m,m/2)$-functions
in the form of $F(x,y)$, where $(x,y)\in\F_{2^{m/2}}\times\F_{2^{m/2}}$.
\begin{enumerate}
  \item [1.] $F(x,y)=L(x\pi(y))+H(y)$ \cite{Nyberg1991PN}, where the product $x\pi(y)$ is calculated in $\F_{2^{m/2}}$,
   $L$ is any linear or affine mapping from $\F_{2^{m/2}}$ onto itself,
   $\pi$ is any permutation of $\F_{2^{m/2}}$ and $H$ is any $(m/2,m/2)$-function.
   This class of functions are called strict Maiorana-McFarland class.
   Maiorana-McFarland class of bent vectorial functions can be extended to
   a more general class which is called general Maiorana-McFarland class, see \cite{SIKvecbent99,carlet2010vectorial}.
  \item [2.] $F(x,y)=G(xy^{2^{m/2}-2})$ \cite{carlet2010vectorial}, where $G$ is a balanced $(m/2,m/2)$-function.
  The component functions of $F$ belongs to the class of $\mathcal{PS}_{ap}$ functions \cite{Dil74}.

  \item [3.] $F(x,y)=xG(yx^{2^m-2})$ \cite{Mesnager13DCC}, where $G$ is an o-polynomial on $\F_{2^{m/2}}$.
              This class bent vectorial functions are called $\mathcal{H}$ class in \cite{Mesnager13DCC}.
\end{enumerate}
There are also some primary constructions of bent vectorial $(m,m/2)$-functions
from single term or multiple terms trace functions, see for examples in \cite{XuWuBent15,MPBvecbentIT14}.
\begin{remark}\label{R:PN}
When $m$ is small, some linear codes introduced in Corollary \ref{C:codeVBFallPN}
are optimal or at least have the same parameters as the best known codes listed in \cite{Grassl:codetables}.
For examples:
\begin{itemize}
  \item For $m=6$, we define $F(x,y)=(Tr_1^3(xy),Tr_1^3(\alpha xy),Tr_1^3(\alpha^2 xy))$,
  where $x,y\in\F_{2^3}$ and $\alpha$ is the default primitive element of $\F_{2^3}$ in Magma version 2.12-16.
  Let $f_\lambda=Tr_1^3(xy)$. Then $\mathcal{C}_{D_{f_\lambda}}$ defined by
  \eqref{E:codeVBFall} is a $[28,9,10]$-code with weight enumerator
  $1+84z^{10}+63z^{12}+216z^{14}+63z^{16}+84z^{18}+z^{28}$ from our Magma program, which is an optimal code \cite{Grassl:codetables} and confirms the result of Corollary \ref{C:codeVBFallPN}.
  If we define $F(x,y)=(Tr_1^3(xy)+1,Tr_1^3(\alpha xy),Tr_1^3(\alpha^2 xy))$ and $f_\lambda=Tr_1^3(xy)+1$.
  Then $\mathcal{C}_{D_{f_\lambda}}$ defined by \eqref{E:codeVBFall} is a $[36,9,14]$-code with weight enumerator
  $1+108z^{14}+63z^{16}+168z^{18}+63z^{20}+108z^{22}+z^{36}$ by Magma programm, which is an optimal code and confirms the result of Corollary \ref{C:codeVBFallPN}.

  \item For $m=8$, we define $F(x,y)=(Tr_1^4(xy),Tr_1^4(\alpha xy),Tr_1^4(\alpha^2 xy),Tr_1^4(\alpha^3 xy))$,
  where $x,y\in\F_{2^4}$ and $\alpha$ is the default primitive element of $\F_{2^4}$ in Magma version 2.12-16.
  Let $f_\lambda=Tr_1^4(xy)$. Then the linear code $\mathcal{C}_{D_{f_\lambda}}$ defined by \eqref{E:codeVBFall}
  is a $[120,12,52]$-code with weight enumerator $1+840z^{52}+255z^{56}+1904z^{60}+255z^{64}+840z^{68}+z^{120}$ according to our Magma program,
  which confirms the result of Corollary \ref{C:codeVBFallPN}.
  This code has the same parameters as a best known linear code given in \cite{Grassl:codetables}.
  If we define $F(x,y)=(Tr_1^4(xy)+1,Tr_1^4(\alpha xy),Tr_1^4(\alpha^2 xy),Tr_1^4(\alpha^3 xy))$ and
  $f_\lambda=Tr_1^4(xy)+1$. Then the linear code $\mathcal{C}_{D_{f_\lambda}}$ defined by \eqref{E:codeVBFall}
  is a $[136,12,60]$-code with weight enumerator $1+952z^{60}+255z^{64}+1680z^{68}+255z^{72}+952z^{76}+z^{136}$ by our Magma program,
  which confirms the result of Corollary \ref{C:codeVBFallPN}.
  This code has the same parameters as a best known linear code given in \cite{Grassl:codetables}.
\end{itemize}
\end{remark}

If $F$ is an almost bent function from $\F_{2^m}$ to itself, then by definition, all values in the extended Walsh spectrum of $F$
belong to the set $\{0,2^{(m+1)/2}\}$.
\begin{corollary}\label{C:codeVBFallAB}
Let $F$ be an almost bent function from $\F_{2^m}$ to itself.
Then there esixts three linear codes with parameters $[2^{m-1}-2^{(m-1)/2}, 2m, 2^{m-2}-3\cdot 2^{(m-3)/2}]$,
$[2^{m-1}+2^{(m-1)/2}, 2m, 2^{m-2}-2^{(m-3)/2}]$  and
$[2^{m-1}, 2m, 2^{m-2}-2^{(m-1)/2}]$ respectively.
\end{corollary}
We list the known power almost bent functions $F(x)=x^d$ on $\F_{2^m}$ in the following:
\begin{enumerate}
  \item $d=2^i+1$, where $\gcd(m,i)=1$ is odd \cite{Gold1968IT}.
  These power functions are called Gold functions.

  \item $d=2^{2i}-2^i+1$, where $i\geq 2\leq (m-1)/2$ and $\gcd(m,i)=1$.
   The AB property of this function is equivalent to a result given by Kasami \cite{kasami1971};
   Welch also obtained this result but never published it.
   These power functions are called Kasami functions or Kasami-Welch functions.

  \item $d=2^{(m-1)/2}+3$. These power functions were conjectured AB by Welch and this was proved by Canteaut, Charpin
  and Dobbertin in \cite{CCDWelch2000}.

  \item $d=2^{(m-1)/2}+2^{(m-1)/4}-1$, where $m\equiv 1 \mod 4$.
  These power functions were conjectured AB by Niho, and this was proved by Hollman and Xiang in \cite{HXNihoproof}.

  \item $d=2^{(m-1)/2}+2^{(3m-1)/4}-1$, where $m\equiv 3 \mod 4$.
  These power functions were also conjectured AB by Niho, and  this was  proved by Hollman and Xiang in \cite{HXNihoproof}.
  The power functions in these two last cases are called Niho functions.

\end{enumerate}

\begin{remark}
For small number of $m$, some linear codes introduced in Corollary \ref{C:codeVBFallAB} are optimal or the same as the best known codes \cite{Grassl:codetables}.
For examples:
\begin{enumerate}
  \item For $m=5$, we define $F(x)=x^3$ on $\F_{2^5}$, in which
  the finite field $\F_{2^5}$ generated by the default
  primitive polynomial $x^{5 }+ x^{2 }+ 1$ in Magma version 2.12-16.
  Let $f_{\lambda}$ be the function $Tr_1^5(x^3+\alpha^3x)$, $Tr_1^5(x^3)$, $Tr_1^5(x^3+x)$ respectively, where
  $\alpha$ is a root of the equation $x^{5 }+ x^{2 }+ 1=0$.
  Then the linear codes $\mathcal{C}_{D_{f_\lambda}}$ defined by \eqref{E:codeVBFall}
  are code $[12,10,2]$-code with weight enumerator $1+30z^2+255z^4+452z^6+255z^8+30z^{10}+z^{12}$,
  $[16,10,4]$-code with weight enumerator $1+60z^4+256z^6+390z^8+256z^{10}+60z^{12}+z^{16}$
  and
  $[20,10,6]$-code with weight enumerator $1+90z^6+255z^8+332z^{10}+255z^{12}+90z^{14}+z^{20}$, respectively,
  by Magma programs. This confirms the results of Corollary \ref{C:codeVBFallAB}.
  These three codes are optimal \cite{Grassl:codetables}.

  \item For $m=7$, we define $F(x)=x^3$ on $\F_{2^7}$, in which
  the finite field $\F_{2^7}$ generated by the default
  primitive polynomial $x^{7}+ x + 1$ in Magma version 2.12-16.
  Let $f_{\lambda}$ be the function $Tr_1^7(x^3+\alpha^7x)$, $Tr_1^7(x^3)$, $Tr_1^7(x^3+\alpha^{19}x)$ respectively, where
  $\alpha$ is a root of the equation $x^{7}+x+1=0$.
  Thus, the linear codes $\mathcal{C}_{D_{f_\lambda}}$ defined by \eqref{E:codeVBFall}
  are code $[56,14,20]$-code with weight enumerator $1+756z^{20}+4095z^{24}+6680z^{28}+4095z^{32}+756z^{36}+z^{56}$,
  $[64,14,24]$-code with weight enumerator $1+1008z^{24}+4096z^{28}+6174z^{32}+4096z^{36}+1008z^{40}+z^{64}$
  and
  $[72,14,28]$-code with weight enumerator $1+1260z^{28}+4095z^{32}+5672z^{36}+4095z^{40}+1260z^{44}+z^{72}$, respectively,
  by Magma programs. This confirms the results of Corollary \ref{C:codeVBFallAB}.
  These codes are the same as the best known codes with such parameters and
  are almost optimal because the upper bounds on the minimum Hamming weight of length $56,72,64$ with
  dimension $14$ are $21,29,25$ respectively \cite{Grassl:codetables}.

  \item For $n=9$, we define $F(x)=x^3$ on $\F_{2^9}$, in which
  the finite field $\F_{2^9}$ generated by the default
  primitive polynomial $x^{9 }+ x^{4 }+1$ in Magma version 2.12-16.
  Let $f_{\lambda}$ be the function $Tr_1^9(x^3+\alpha^9x)$, $Tr_1^9(x^3)$, $Tr_1^9(x^3+\alpha^{10}x)$ respectively, where
  $\alpha$ is a root of the equation $x^{9}+x^4+1=0$.
  Thus, the linear codes $\mathcal{C}_{D_{f_\lambda}}$ defined by \eqref{E:codeVBFall}
  are code $[240,18,104]$-code with weight enumerator $1+14280z^{104}+65535z^{112}+102512z^{120}+65535z^{128}+14280z^{136}+z^{240}$,
  $[256,18,112]$-code with weight enumerator $1+16320z^{112}+65536z^{120}+98430z^{128}+65536z^{136}+16320z^{144}+z^{256}$
  and
  $[272,18,120]$-code with weight enumerator $1+18360z^{120}+65535z^{128}+94352z^{136}+65535z^{144}+18360z^{152}+z^{272}$, respectively,
  by Magma programs. This confirms the results of Corollary \ref{C:codeVBFallAB}.
  The fist two codes are the same as the best known codes.
\end{enumerate}
\end{remark}

\subsection{The weight distribution of $\mathcal{C}_{D_{f_\lambda}}$ when
$F$ is a perfect nonlinear function}

Let $C$ be a binary linear $[n,k,d]$-code including the all-one codeword.
Then the number $A_{w_1}$ of codewords with Hamming weight $w_1$
is equal to the number $A_{w_2}$ of codewords with Hamming weight $w_2=n-w_1$.

\begin{theorem}\label{T:wdPN0all}
Let $F$ be a perfect nonlinear function from $\F_{2^m}$ to $\F_{2^{m/2}}$, where $m$ is even.
For every $\lambda\in\F_{2^{m/2}}^*$, $\mathcal{C}_{D_{f_\lambda}}$ is an
$[n_{f_\lambda},3m/2,n_{f_\lambda}/2-2^{m/2-1}]$-code
with the weight distribution given in Table \ref{Table:wdPN0all},
where $f_\lambda=Tr_1^{m/2}(\lambda F)$ and
$n_{f_\lambda}\in\{2^{m-1}-2^{m/2-1},2^{m-1}+2^{m/2-1}\}$.

\begin{table}[h!]
    \begin{center}
    \caption{The weight distribution of the code of Theorem \ref{T:wdPN0all}}\label{Table:wdPN0all}
    \begin{tabular}{|c|c|c|c|c|c|c|c|c|c|c|c|c|c|c|c}
    \hline  Weight $w$ & Multiplicity $A_w$
    \\\hline   $0$  &$1$
    \\\hline   $\frac{n_{f_\lambda}}{2}-2^{\frac m2-1}$   &$2^{m/2-1}n_{f_\lambda}-2^{-m}n_{f_\lambda}^2-2^{m-2}+\frac{1}{4}$
    \\\hline   $\frac{n_{f_\lambda}}{2}-2^{\frac m2-2}$   &$2^m-1$
    \\\hline   $\frac{n_{f_\lambda}}{2}$                  &$(2n_{f_\lambda}^2-2^{3m/2}n_{f_\lambda})2^{-m}+2^{3m/2}-3\cdot 2^{m-1}-\frac{1}{2}$
    \\\hline   $\frac{n_{f_\lambda}}{2}+2^{\frac m2-2}$   &$2^m-1$
    \\\hline   $\frac{n_{f_\lambda}}{2}+2^{\frac m2-1}$   &$2^{m/2-1}n_{f_\lambda}-2^{-m}n_{f_\lambda}^2-2^{m-2}+\frac{1}{4}$
   \\\hline    $n_{f_\lambda}$ &$1$
    \\ \hline
    \end{tabular}
    \end{center}
    \end{table}
\end{theorem}
\begin{proof}
Note that for any $\lambda\in\F_{2^{m/2}}^*$ the Boolean function $f_\lambda$
is a bent function and hence $\widehat{f_\lambda}(a)=\pm 2^{m/2}$ for all $a\in\F_{2^m}$.
Thus, we have $n_{f_\lambda}=|\{x\in\F_{2^m}: f_\lambda(x)=1\}|\in\{2^{m-1}-2^{n/2-1},2^{m-1}+2^{n/2-1}\}$.
Note that $F$ has nonlinearity $2^{m-1}-2^{m/2-1}$.
By Proposition \ref{T:codeVBFall}, we immediately obtain that
$\mathcal{C}_{D_{f_\lambda}}$ is an $[n_{f_\lambda},3m/2,n_{f_\lambda}/2-2^{m/2-1}]$-code.
In what follows, we discuss the weight distribution of $\mathcal{C}_{D_{f_\lambda}}$.
Since the code $\mathcal{C}_{D_{f_\lambda}}$ is linear, the Hamming wights of all nonzero codewords of $\mathcal{C}_{D_{f_\lambda}}$
belong to the set constituted by the values of Hamming distances between all pairs of codewords
in $\mathcal{C}_{D_{f_\lambda}}$.
Thus, by \eqref{E:pairsdHsetlambda0} and \eqref{E:pairsdHsetnotlambda0} we can see that the
Hamming wights of all codewords of $\mathcal{C}_{D_{f_\lambda}}$
belong to the set $\{1/2n_{f_\lambda}\pm 2^{m/2-2}, 1/2n_{f_\lambda}\pm 2^{m/2-1},1/2n_{f_\lambda},n_{f_\lambda}\}\cup\{0\}$.
So we can assume that the codewords in $\mathcal{C}_{D_{f_\lambda}}$ have weights:
$$w_1=0,w_2=\frac{n_{f_\lambda}-2^{m/2}}{2},w_3=\frac{n_{f_\lambda}-2^{m/2-1}}{2},$$
$$w_4=\frac{n_{f_\lambda}}{2},w_5=\frac{n_{f_\lambda}+2^{m/2-1}}{2},w_6=\frac{n_{f_\lambda}+2^{m/2}}{2},w_7=n_{f_\lambda},$$
since the all values in the extended Walsh spectrum of $F$ are equal to $2^{m/2}$.
We now determine the number $A_{w_i}$ of codewords with weight
$w_i$ in $\mathcal{C}_{D_{f_\lambda}}$, where $i=1,2,\cdots,7$.
Clearly we have $A_{w_1}=A_{w_7}=1$.
Furthermore, it follows from \eqref{E:pairsdHsetlambda0} and \eqref{E:pairsdHsetnotlambda0}
that the Hamming weight of any nonzero codeword $c_{x,y}$ belongs to the set
$\{1/2n_{f_\lambda}\pm 2^{m/2-2},n_{f_\lambda}\}$ if and only if $y\in\{0,\lambda\}$.
Note that the set $\{c_{x,y}: x\in\F_{2^m},y\in\{0,\lambda\}\}$ includes the
all-one and all-zero vectors. So we have $w_3+w_5=|\{c_{x,y}: x\in\F_{2^m},y\in\{0,\lambda\}\}|-2=2^{m+1}-2$
and then $A_{w_3}=A_{w_5}=2^m-1$ by the observation made before the present theorem.
Note that any two elements $d_i$ and $d_j$ of the set $D_{f_\lambda}$ must be distinct if $i\not=j$.
Then by \cite[Theorem 10]{MS1977}, the minimum weight of the dual code of $\mathcal{C}_{D_{f_\lambda}}$ is no less than $3$.
According to the first three Pless Power Moments \cite[p. 260]{HPbook2003}, we have
\begin{eqnarray*}
\left\{\begin{array}{lll}
A_{w_1}+A_{w_2}+A_{w_3}+A_{w_4}+A_{w_5}+A_{w_6}+A_{w_7}&=&2^{3m/2}\\
w_1A_{w_1}+w_2A_{w_2}+w_3A_{w_3}+w_4A_{w_4}+w_5A_{w_5}+w_6A_{w_6}+w_7A_{w_7}&=&n_{f_\lambda}2^{3m/2-1}\\
w_1^2A_{w_1}+w_2^2A_{w_2}+w_3^2A_{w_3}+w_4^2A_{w_4}+w_5^2A_{w_5}+w_6^2A_{w_6}+w_7^2A_{w_7}&=&n_{f_\lambda}(n_{f_\lambda}+1)2^{3m/2-2}
\end{array} \right..
\end{eqnarray*}
Recall that $A_{w_1}=A_{w_7}=1$ and $A_{w_3}=A_{w_5}=2^m-1$ and note that $A_{w_2}=A_{w_6}$. Thus
\begin{eqnarray*}
\left\{\begin{array}{lll}
2A_{w_2}+A_{w_4}&=&2^{3m/2}-2^{m+1}\\
(w_2^2+w_6^2)A_{w_2}+w_4^2A_{w_4}&=&n_{f_\lambda}(n_{f_\lambda}+1)2^{3m/2-2}-n_{f_\lambda}^2-
(n_{f_\lambda}^2/2+2^{m-3})(2^m-1)
\end{array} \right..
\end{eqnarray*}
By solving this system of equations, we
can get the weight distribution of $\mathcal{C}_{D_{f_\lambda}}$ given in Table \ref{Table:wdPN0all}.
This completes the proof.
\end{proof}

\subsection{The weight distribution of $\mathcal{C}_{D_{f_\lambda}}$ when $F$ is an almost bent functions}

\begin{lemma}[\cite{AnneDecomposingbent}]\label{L:Anne}
Let $f$ be a Boolean function of $m$ variables such that its Walsh spectrum takes at most three
values $0,\pm 2^s$. Then $|\{a\in\F_2^m: \widehat{f}(a)=0\}|=2^m-2^{2m-2s}$,
$|\{a\in\F_2^m: \widehat{f}(a)=2^s\}|=2^{2m-2s-1}+(-1)^{f(0)}2^{m-s-1}$,
and $|\{a\in\F_2^m: \widehat{f}(a)=-2^s\}|=2^{2m-2s-1}-(-1)^{f(0)}2^{m-s-1}$.
\end{lemma}

\begin{lemma}\label{L:wtdistriAB0all}
Let $F(x)=x^d$ be an almost bent function over $\F_{2^m}$, where $m$ is odd and $\gcd(d,2^m-1)=1$.
Define
$N_{f_\lambda}(t)=|\{(x,y)\in\F_{2^m}\times\F_{2^m}\setminus\{0,\lambda\}: \widehat{f}_{\lambda+y}(x)-\widehat{f}_y(x)=t\}|$,
where $\lambda\in\F_{2^m}^*$ and $f_\lambda=Tr_1^{m}(\lambda F)$.
Then we have $N_{f_\lambda}(2^{(m+3)/2})=2^{2m-4}-2^{m-3}$,
$N_{f_\lambda}(-2^{(m+3)/2})=2^{2m-4}-2^{m-3}$,
$N_{f_\lambda}(2^{(m+1)/2})=2^{2m-2}-2^{m-1}$,
$N_{f_\lambda}(-2^{(m+1)/2})=2^{2m-2}-2^{m-1}$, $N_{f_\lambda}(0)=3\cdot2^{2m-3}+2^{m-2}-2^m$ and $N_{f_\lambda}(t)=0$ for every other value of $t$.
\end{lemma}
\begin{proof}
Note that
\begin{eqnarray*}
&&\widehat{f}_{\lambda+y}(x)-\widehat{f}_y(x)\\
&=&\sum_{z\in\F_{2^m}}(-1)^{Tr_1^m((\lambda+y)z^d+xz)}-\sum_{z\in\F_{2^m}}(-1)^{Tr_1^m(yz^d+xz)}\\
&=&\sum_{z\in\F_{2^m}}(-1)^{Tr_1^m(z^d+x(\lambda+y)^{-\frac1d}z)}-\sum_{z\in\F_{2^m}}(-1)^{Tr_1^m(z^d+xy^{-\frac1d}z)}\\
&=&\sum_{z\in\F_{2^m}}(-1)^{Tr_1^m(z^d+wz)}-\sum_{z\in\F_{2^m}}(-1)^{Tr_1^m(z^d+w(1+\frac{\lambda}{y})^{\frac1d}z)}\\
&=&\widehat{f_1}(w)-\widehat{f_1}{\Big (}w{\Big (}1+\frac{\lambda}{y}{\Big )}^{\frac1d}{\Big )}
\end{eqnarray*}
where $w=x(\lambda+y)^{-\frac1d}$.
Then we can see that
$N_{f_\lambda}(t)=|\{(w,y)\in\F_{2^m}\times\F_{2^m}\setminus\{0,\lambda\}: \widehat{f}_{1}(w)-\widehat{f}_1(w(1+{\lambda}/{y})^{1/d})=t\}|$.
Note that for any $w\in\F_{2^m}^*$, $w(1+{\lambda}/{y})^{1/d}$ ranges over $\F_{2^m}\setminus\{0,w\}$ when
$y$ ranges over $\F_{2^m}\setminus\{0,\lambda\}$.
Note also that $|\{w\in\F_{2^m}:\widehat{f}_{\lambda}(w)=0\}|=2^{m-1}$,
$|\{w\in\F_{2^m}:\widehat{f}_{\lambda}(w)=2^{(m+1)/2}\}|=2^{m-2}+2^{(m+3)/2}$
and $|\{w\in\F_{2^m}:\widehat{f}_{\lambda}(w)=-2^{(m+1)/2}\}|=2^{m-2}-2^{(m+3)/2}$,
according to Lemma \ref{L:Anne} and $f_{\lambda}(0)=0$.
So we have
$N_{f_\lambda}(2^{(m+3)/2})=(2^{m-2}+2^{(m-3)/2})(2^{m-2}-2^{(m-3)/2})=2^{2m-4}-2^{m-3}$,
$N_{f_\lambda}(-2^{(m+3)/2})=(2^{m-2}-2^{(m-3)/2})(2^{m-2}+2^{(m-3)/2})=2^{2m-4}-2^{m-3}$,
$N_{f_\lambda}(2^{(m+1)/2})=(2^{m-2}+2^{(m-3)/2}-1)(2^{m-1})+2^{m-1}(2^{m-2}-2^{(m-3)/2})=2^{2m-2}-2^{m-1}$,
$N_{f_\lambda}(-2^{(m+1)/2})=(2^{m-2}-2^{(m+3)/2}-1)(2^{m-1})+2^{m-1}(2^{m-2}+2^{(m+3)/2})=2^{2m-2}-2^{m-1}$,
and $N_{f_\lambda}(0)=2^m(2^m-2)-N_{f_\lambda}(2^{(m+3)/2})-N_{f_\lambda}(-2^{(m+3)/2})
-N_{f_\lambda}(2^{(m+1)/2})-N_{f_\lambda}(-2^{(m+1)/2})=3\cdot2^{2m-3}+2^{m-2}-2^m$.
This completes the proof.
\end{proof}

\begin{theorem}\label{T:wdAB0all}
Let $F(x)=x^d$ be an almost bent function over $\F_{2^m}$, where $m$ is odd and $\gcd(d,2^m-1)=1$.
For every $\lambda\in\F_{2^{m/2}}^*$, $\mathcal{C}_{D_{f_\lambda}}$ is a
$[2^{m-1},2m, 2^{m-2}-2^{(m-1)/2}]$-code with the weight distribution
given in Table \ref{Table:wdAB0all}, where $f_\lambda=Tr_1^{m}(\lambda F)$.
\begin{table}[h!]
\begin{center}
\caption{The weight distribution of the code of Theorem \ref{T:wdAB0all}}\label{Table:wdAB0all}
\begin{tabular}{|c|c|c|c|c|c|c|c|c|c|c|c|c|c|c|c}
    \hline  Weight $w$ & Multiplicity $A_w$
    \\\hline   $0$   &$1$
    \\\hline   $2^{m-2}-2^{(m-1)/2}$   &$2^{2m-4}-2^{m-3}$
    \\\hline   $2^{m-2}-2^{(m-3)/2}$   &$2^{2m-2}$
    \\\hline   $2^{m-2}$               &$3\cdot2^{2m-3}+2^{m-2}-2$
    \\\hline   $2^{m-2}+2^{(m-3)/2}$   &$2^{2m-2}$
    \\\hline   $2^{m-2}+2^{(m-1)/2}$   &$2^{2m-4}-2^{m-3}$
    \\\hline   $2^{m-1}$   &$1$
    \\ \hline
\end{tabular}
\end{center}
\end{table}
\end{theorem}

\begin{proof}
It is clear that for any $\lambda\in\F_{2^m}^*$ the Boolean function $f_\lambda$
is a semi-bent function and hence $\widehat{f_\lambda}(a)\in\{0,\pm 2^{(m+1)/2}\}$ for all $a\in\F_{2^m}$
and $nl(F)=2^{m-1}-2^{(m-1)/2}$.
Note that $F$ is bijective. We have $f_\lambda$ is balanced for any $\lambda\in\F_{2^s}^*$ \cite{carlet2010vectorial}
and hence $n_{f_\lambda}=2^{m-1}$. Hence $\mathcal{C}_{D_{f_\lambda}}$ is an
$[n_{f_\lambda},3m/2,n_{f_\lambda}/2-2^{m/2-1}]$-code by Proposition \ref{T:codeVBFall}.

In what follows we determine the weight distribution of $\mathcal{C}_{D_{f_\lambda}}$.
Let us consider the Hamming weight of $c_{x,y}$ for any $(x,y)\in\F_{2^m}\times\F_{2^m}$,
where $c_{x,y}$ is defined in \eqref{E:codeVBFall}.
Recall that the all values in the extended Walsh spectrum of $F$ are belong
to the set $\{0,2^{(m+1)/2}\}$, then by \eqref{E:pairsdHsetlambda0} and \eqref{E:pairsdHsetnotlambda0} we can see that
the codewords in $\mathcal{C}_{D_{f_\lambda}}$ have weights:
$$w_1=0,w_2=2^{m-1}-2^{\frac{m-1}{2}},w_3=2^{m-1}-2^{\frac{m-3}{2}},$$
$$w_4=2^{m-1},w_5=2^{m-1}+2^{\frac{m-3}{2}},w_6=2^{m-1}+2^{\frac{m-1}{2}},w_7=n_{f_\lambda}.$$
We now determine the number $A_{w_i}$ of codewords with weight $w_i$ in $\mathcal{C}_{D_{f_\lambda}}$,
where $i=1,2,\cdots,7$. Particularly, we have $A_{w_1}=A_{w_7}=1$, $A_{w_2}=A_{w_6}$ and $A_{w_3}=A_{w_5}$,
since this linear code includes the all-one codeword.
It can be seen that the codewords $c_{x,y}$ with weights $w_2$ and $w_6$ only appear in the case of
$(x,y)\in\F_{2^m}\times\F_{2^m}\setminus\{0,\lambda\}$, according to Case 3 of the proof of Proposition \ref{T:codeVBFall}.
Then by Lemma \ref{L:wtdistriAB0all}, we immediately get $w_2=w_6=2^{2m-4}-2^{m-3}$.
We can easily see that the codewords $c_{x,y}$ with weights $w_3$ and $w_5$ only appear in the
case of $(x,y)\in\F_{2^m}^*\times\{0,\lambda\} \cup \F_{2^m}\times\F_{2^m}\setminus\{0,\lambda\}$,
by Cases 1,2,3 of the proof of Proposition \ref{T:codeVBFall}.
Note that $|\{w\in\F_{2^m}:\widehat{f}_{\lambda}(w)=0\}|=2^{m-1}$,
$|\{w\in\F_{2^m}:\widehat{f}_{\lambda}(w)=2^{(m+1)/2}\}|=2^{m-2}+2^{(m+3)/2}$,
and $|\{w\in\F_{2^m}:\widehat{f}_{\lambda}(w)=-2^{(m+1)/2}\}|=2^{m-2}-2^{(m+3)/2}$, according to Lemma \ref{L:Anne},
and that $\widehat{f}_\lambda(0)=0$.
Thus we can get that the sum of the numbers of codewords $c_{x,y}$ with weights $w_3$ and $w_5$
when $(x,y)\in\F_{2^m}^*\times\{0,\lambda\}$ is equal to $2^{m}$.
So, by Lemma \ref{L:wtdistriAB0all}, we have $A_{w_3}+A_{w_5}=2(2^{2m-2}-2^{m-1})+2^m=2^{2m-1}$ and hence $A_{w_3}=A_{w_5}=2^{2m-2}$.
Then we have $A_{w_4}=2^{2m}-2\sum_{i=1}^3A_{w_i}=3\cdot2^{2m-3}+2^{m-2}-2$.
This completes the proof.
\end{proof}

\section{The weight distribution of subcodes}\label{sec:subcode}

The linear code $\mathcal{C}_{D_{f_\lambda}}$ defined by \eqref{E:codeVBFall} includes
the all-one codeword. In this sections, we focus on calculating the weight distribution of those subcodes
of $\mathcal{C}_{D_{f_\lambda}}$ which do not contain the all-one codewords.

For any $\lambda\in\F_{2^s}^*$, we denote by $H_\lambda$ a set such that:
(1) $H_\lambda$ is a vector subspace of $\F_2^s$ with dimension $s-1$; and
(2) $H_{\lambda}\cup \{\lambda+H_{\lambda}\}=\F_{2^s}$, \textit{i.e.},
$\F_{2^s}=\{x+y:x\in\{0,\lambda\}, y\in H_\lambda\}$.
Let $m,s$ be two integers and $F$ be a vectorial Boolean function from $\F_{2^m}$ to $\F_{2^s}$.
For any $\lambda\in\F_{2^s}^*$, we denote by $f_\lambda$ the Boolean function $Tr_1^s(\lambda F)$
and define a linear code of length $n_{f_\lambda}$ over $\F_2$ as follows:
\begin{eqnarray}\label{E:codeVBF}
\mathcal{C}_{D_{f_\lambda},H_\lambda}=\{c_{x,y}: x\in \F_{2^m},y\in H_{\lambda}\},
\end{eqnarray}
where $n_{f_\lambda}=|D_{f_\lambda}|=|\{d\in\F_{2^m}:f_\lambda(d)\not=0\}|=|\{d_1,d_2,\cdots d_{n_{f_\lambda}}\}|$
and $c_{x,y}=\big(Tr_1^m(xd_1)+Tr_1^s(yF(d_1)),\cdots,Tr_1^m(xd_n)+Tr_1^s(yF(d_{n_{f_\lambda}}))\big).$

\vspace{5mm}
Let $F$ be an $(m,s)$-function, where $m,s$ be two integers.
It can be easily checked that, for any $\lambda\in\F_{2^s}^*$,
the linear code $\mathcal{C}_{D_{f_\lambda},H_\lambda}$ defined by \eqref{E:codeVBF}
is a subcode of $\mathcal{C}_{D_{f_\lambda}}$ defined by \eqref{E:codeVBFall}.
By Proposition \ref{T:codeVBFall}, we directly have the following theorem.

\begin{theorem}\label{T:codeVBF}
Let $F$ be an $(m,s)$-function. For any $\lambda\in\F_{2^m}^*$, if $2^m-2nl(F)<n_{f_\lambda}$,
where $f_\lambda=Tr_1^s(\lambda F)$, denoting by $n_{f_\lambda}$ the size of the support of $f_\lambda$,
 the linear code $\mathcal{C}_{D_{f_\lambda},H_\lambda}$ defined by \eqref{E:codeVBF} has length $n_{f_\lambda}$,
dimension $m+s-1$ and minimum Hamming weight no less than $nl(F)-\frac{2^m-n_{f_\lambda}}{2}$.
\end{theorem}

\subsection{The weight distribution of $\mathcal{C}_{D_{f_\lambda},H_\lambda}$ when
$F$ is a perfect nonlinear function}

\begin{theorem}\label{T:wdPN0}
Let $F$ be a perfect nonlinear function from $\F_{2^m}$, where $m$ is even, to $\F_{2^{m/2}}$ such that $F(0)=0$.
For every $\lambda\in\F_{2^m}^*$ and any $H_\lambda$,
$\mathcal{C}_{D_{f_\lambda},H_\lambda}$ is an $[n_{f_\lambda},3m/2-1,n_{f_\lambda}/2-2^{m/2-1}]$-code
with the weight distribution given in Table \ref{Table:wdPN0},
where $f_\lambda=Tr_1^{m/2}(\lambda F)$ and $n_{f_\lambda}\in\{2^{m-1}-2^{m/2-1},2^{m-1}+2^{m/2-1}\}$
is defined as the size of the support of $f_\lambda$.
   \begin{table}[h!]
    \begin{center}
    \caption{The weight distribution of the code of Theorem \ref{T:wdPN0}}\label{Table:wdPN0}
    \begin{tabular}{|c|c|c|c|c|c|c|c|c|c|c|c|c|c|c|c}
    \hline  Weight $w$ & Multiplicity $A_w$
    \\\hline   $0$  &$1$
    \\\hline   $\frac{n_{f}}{2}-2^{\frac m2-1}$   &$2^{m/2-2}n_{f_\lambda}-2^{-m-1}n_{f_\lambda}^2-2^{m-3}+\frac{1}{8}$
    \\\hline   $\frac{n_{f}}{2}-2^{\frac m2-2}$   &$\frac{2^m-1-n_{f_\lambda}2^{-\frac{m-2}{2}}}{2}$
    \\\hline   $\frac{n_{f}}{2}$                  &$(n_{f_\lambda}^2-2^{3m/2-1}n_{f_\lambda})2^{-m}+2^{3m/2-1}-3\cdot 2^{m-2}-\frac{1}{4}$
    \\\hline   $\frac{n_{f}}{2}+2^{\frac m2-2}$   &$\frac{2^m-1+n_{f_\lambda}2^{-\frac{m-2}{2}}}{2}$
    \\\hline   $\frac{n_{f}}{2}+2^{\frac m2-1}$   &$2^{m/2-2}n_{f_\lambda}-2^{-m-1}n_{f_\lambda}^2-2^{m-3}+\frac{1}{8}$
    \\ \hline
    \end{tabular}
    \end{center}
    \end{table}
\end{theorem}
\begin{proof}
Recall that for any $\lambda\in\F_{2^{m/2}}^*$ the Boolean function $f_\lambda$
is a bent function and so $\widehat{f_\lambda}(a)=\pm 2^{m/2}$ for all $a\in\F_{2^m}$.
This implies that $n_{f_\lambda}=|\{x\in\F_{2^m}: f_\lambda(x)=1\}|\in\{2^{m-1}-2^{n/2-1},2^{m-1}+2^{n/2-1}\}$.
Recall also that $F$ has nonlinearity $2^{m-1}-2^{m/2-1}$.
By Theorem \ref{T:codeVBF} we have that $\mathcal{C}_{D_{f_\lambda},H_\lambda}$ is
a linear code with length $n_{f_\lambda}$ and dimension $3m/2-1$.

In the rest of this proof, we determine the weight distribution of $\mathcal{C}_{D_{f_\lambda},H_\lambda}$.
It can be easily seen that the code $\mathcal{C}_{D_{f_\lambda},H_\lambda}$ does not include the all-one codeword.
Thus, by recalling the proof of Theorem \ref{T:wdPN0all},
we can assume that the codewords in $\mathcal{C}_{D_{f_\lambda},H_\lambda}$ have weights:
$$w_1=0,w_2=\frac{n_{f_\lambda}-2^{m/2}}{2},w_3=\frac{n_{f_\lambda}-2^{m/2-1}}{2},$$
$$w_4=\frac{n_{f_\lambda}}{2},w_5=\frac{n_{f_\lambda}+2^{m/2-1}}{2},w_6=\frac{n_{f_\lambda}+2^{m/2}}{2}.$$
We now determine the number $A_{w_i}$ of codewords with weight
$w_i$ in $\mathcal{C}_{D_{f_\lambda},H_\lambda}$, where $i=1,2,\cdots,6$. Obviously we have $A_{w_1}=1$.
By \eqref{E:pairsdHsetlambda0} and \eqref{E:pairsdHsetnotlambda0},
the Hamming weight of any nonzero codeword $c_{x,y}$, where $c_{x,y}$ is defined by
\ref{E:codeVBF}, belongs to the set $\{1/2n_{f_\lambda}\pm 2^{m/2-2},n_{f_\lambda}\}$ if
and only if $y\in\{0,\lambda\}$. Note that $\lambda\not\in H_{\lambda}$; the Hamming weight of any nonzero codeword $c_{x,y}$ belongs to the set
$\{1/2n_{f_\lambda}\pm 2^{m/2-2},n_{f_\lambda}\}$ if and only if $y=0$.
Define $\mathcal{C}'=\{c_{x,0}\in\mathcal{C}_{D_{f_\lambda},H_\lambda}: x\in\F_{2^m}\}$.
We can easily see that $\mathcal{C}'$ is a subcode of $\mathcal{C}_{D_{f_\lambda},H_\lambda}$
with length $n_{f_\lambda}$ and dimension $m$.
Moreover, we can see that the codewords in $\mathcal{C}'$ have weights $w_1,w_3,w_5$
and the numbers of codewords in $\mathcal{C}'$ with weight $w_1,w_3,w_5$ respectively
are equal to $A_{w_1},A_{w_3},A_{w_5}$ respectively.
It follows from \cite[Theorem 10]{MS1977} that the minimum weight of the dual code of $\mathcal{C}'$ is no less than $3$
since any two elements $d_i$ and $d_j$ of the set $D_{f_\lambda}$ must be distinct if $i\not=j$.
According to the first two Pless Power Moments \cite[p. 260]{HPbook2003}, we have
\begin{eqnarray*}
\left\{\begin{array}{lll}
A_{w_1}+A_{w_3}+A_{w_5}&=&2^{m}\\
w_1A_{w_1}+w_3A_{w_3}+w_5A_{w_5}&=&n_{f_\lambda}2^{m-1}\\
\end{array} \right..
\end{eqnarray*}
Recall that $A_{w_1}=1$. By solving this system of equations, we have
$A_{w_3}=(2^m-1-n_{f_\lambda}2^{-\frac{m-2}{2}})/2$
and $A_{w_5}=(2^m-1+n_{f_\lambda}2^{-\frac{m-2}{2}})/2$.
By \cite[Theorem 10]{MS1977}, we immediately get that the minimum weight of the dual code of $\mathcal{C}'$ is no less than $3$
since any two elements $d_i$ and $d_j$ of the set $D_{f_\lambda}$ must be distinct if $i\not=j$.
According to the first three Pless Power Moments \cite[p. 260]{HPbook2003}, we have
\begin{eqnarray*}
\left\{\begin{array}{lll}
A_{w_1}+A_{w_2}+A_{w_3}+A_{w_4}+A_{w_5}+A_{w_6}&=&2^{3m/2-1}\\
w_1A_{w_1}+w_2A_{w_2}+w_3A_{w_3}+w_4A_{w_4}+w_5A_{w_5}+w_6A_{w_6}&=&n_{f_\lambda}2^{3m/2-2}\\
w_1^2A_{w_1}+w_2^2A_{w_2}+w_3^2A_{w_3}+w_4^2A_{w_4}+w_5^2A_{w_5}+w_6^2A_{w_6}&=&n_{f_\lambda}(n_{f_\lambda}+1)2^{3m/2-3}
\end{array} \right..
\end{eqnarray*}
Recall that the values $A_{w_1},A_{w_3},A_{w_5}$. Then by solving this system of equations,
we can get the values of $A_{w_4},A_{w_2},A_{w_6}$. This completes the proof.
\end{proof}

\begin{example}\label{E:PNsubcode}
For $m=8$, we define $F(x,y)=(Tr_1^4(xy),Tr_1^4(\alpha xy),Tr_1^4(\alpha^2 xy),Tr_1^4(\alpha^3 xy))$,
where $x,y\in\F_{2^4}$ and $\alpha$ is the default primitive element of $\F_{2^4}$ in Magma version 2.12-16,
and $f_{\lambda}=Tr_1^4(xy)$ with $H_{\lambda}=c_1\alpha+c_2\alpha^2+c_3\alpha^3$ where $(c_1,c_2,c_3)\in\F_2^3$.
With the help of Magma, we can get the linear code $\mathcal{C}_{D_{f_\lambda},H_\lambda}$ defined by
\eqref{E:codeVBF} is a $[120,11,52]$-code with weight enumerator $1+420z^{52}+120z^{56}+952z^{60}+135z^{64}+420z^{68}$,
which confirms the result of Theorem \ref{T:wdPN0}.
\end{example}

\subsection{The weight distribution of $\mathcal{C}_{D_{f_\lambda},H_\lambda}$ when
$F$ is a Gold function}

For odd $m$ and $e=2^i+1$ where $\gcd(i,m)=1$ and $1\leq i\leq (n-1)/2$,
the $(m,m)$-functions $F(x)=x^e$ are
called Gold functions (Gold functions also exist in even dimensions, but we only consider the odd case in
this paper). These functions are almost bent and are bijective.
We now consider the weight distribution of the linear codes from the Gold functions in odd dimensions.
We first need the following lemmas.

\begin{lemma}[\cite{Carlet10}]\label{L:linearfunction-sum-E} 
Let $m,k$ be two integers such that $m\geq k\geq 1$. Let $E$ be a subspace of $\F_2^m$ with dimension $k$
and with orthogonal space $E^\perp=\{x\in\F_{2^m} : \forall y\in E, Tr_1^m(xy)=0\}$.
Then we have that $\sum_{x\in E}(-1)^{Tr_1^m(\alpha x)}$ equals $2^{k}$ if $\alpha\in E^\perp$ and is $0$ otherwise.
\end{lemma}

\begin{lemma}\label{L:walsh-space-quadraticsemibent}
Let $f$ be a balanced quadratic semi-bent function defined on $\F_{2^m}$ and $T$ be the set $\{\alpha\in\F_2^m:\widehat{f}(\alpha)=0\}$.
Then $T$ is a subspace of $\F_2^m$ with dimension $m-1$. If we assume that $T^\perp=\{0,w\}$, then
\begin{eqnarray*}
\sum_{x\in\F_2^m}(-1)^{f(x)+f(x+b)}&=&
\left\{\begin{array}{lll}
2^m, &\hbox{if~} b=0\\
-2^m,& \hbox{if~} b=w\\
0, &\hbox{otherwise}
\end{array} \right..
\end{eqnarray*}
\end{lemma}
\begin{proof}
It is easy to see that $\sum_{x\in\F_2^m}(-1)^{f(x)+f(x+b)}\in\{0,\pm 2^m\}$, since $f$ has algebraic degree $2$ and so $f(x)+f(x+b)$
has algebraic degree at most $1$. Note that $\sum_{a\in\F_2^m}{\widehat{f~}}^4(a)=2^m\sum_{b\in\F_2^m}\big(\sum_{x\in\F_2^m}(-1)^{f(x)+f(x+b)}\big)^2$, see \cite{Carl93}.
Then by Lemma \ref{L:Anne}, there only exist two elements $\{0,w'\}$, in $\F_2^m$ such that
$|\sum_{x\in\F_2^m}(-1)^{f(x)+f(x+b)}|=2^m$, where $b\in\{0,w'\}$. So we have $\sum_{x\in\F_2^m}(-1)^{f(x)+f(x+b)}=0$
if $b\in \F_2^m\setminus\{0,w'\}$. On the other hand, we have
$\sum_{b\in\F_2^m}\sum_{x\in\F_2^m}(-1)^{f(x)+f(x+b)}=\sum_{x\in\F_2^m}(-1)^{f(x)}\sum_{b\in\F_2^m}(-1)^{f(x+b)}=0$ since $f$ is balanced.
This implies that $\sum_{x\in\F_2^m}(-1)^{f(x)+f(x+w')}=-2^m$.

We now prove that $T$ is a subspace with dimension $m-1$.
Clearly, $0\in T$ since $f$ is balanced. Moreover, $|T|=2^{m-1}$ according to Lemma \ref{L:Anne}.
Hence, for proving that $T$ is a subspace with dimension $m-1$, we only need to prove that
for any two distinct elements $\alpha,\beta\in T$ such that $\widehat{f}(\alpha)=\widehat{f}(\beta)=0$ we have $\widehat{f}(\alpha+\beta)=0$.
Note that $
0={\widehat{f}}^2(\alpha)=\sum_{x,y\in\F_{2^m}}(-1)^{f(x)+f(y)+Tr_1^m(\alpha(x+y))}=\sum_{x,b\in\F_{2^m}}(-1)^{f(x)+f(x+b)+Tr_1^m(\alpha b)}
=\sum_{b\in\F_{2^m}}(-1)^{Tr_1^m(\alpha b)}\sum_{x\in\F_{2^m}}(-1)^{f(x)+f(x+b)}=2^n(1-(-1)^{Tr_1^m(\alpha w')})$,
which implies that $Tr_1^m(\alpha w')=0$. Similarly, we have $Tr_1^m(\beta w')=0$ from $\widehat{f}^2(\beta)$.
Furthermore, we have $\widehat{f}^2(\alpha+\beta)=2^n(2^n(1-(-1)^{Tr_1^m((\alpha+\beta) w')}))$, which is
equal to $0$ since $Tr_1^m((\alpha+\beta)w')=Tr_1^m(\alpha w')+Tr_1^m(\beta w')=0$.
Therefore, $T$ is a subspace with dimension $m-1$.

In what follows, we will prove the rest assertion of this lemma.
From the above discussion, we only need to prove $w=w'$.
Note that
$0=\sum_{\alpha\in T}\widehat{f}^2(\alpha)\\
=\sum_{\alpha\in T}\sum_{x,y\in\F_{2^m}}(-1)^{f(x)+f(y)+Tr_1^m(\alpha(x+y))}
=\sum_{\alpha\in T}\sum_{x,\beta\in\F_{2^m}}(-1)^{f(x)+f(x+\beta)+Tr_1^m(\alpha\beta)}\\
=\sum_{\beta\in\F_{2^m}}\sum_{x\in\F_{2^m}}(-1)^{f(x)+f(x+\beta)} \sum_{\alpha\in T}(-1)^{Tr_1^m(\alpha\beta)}=2^{m-1}\sum_{\beta\in T^\perp}\sum_{x\in\F_{2^m}}(-1)^{f(x)+f(x+\beta)}\\
=2^{2m-1}+2^{m-1}\sum_{x\in\F_{2^m}}(-1)^{f(x)+f(x+w)}$.
So we have $\sum_{x\in\F_{2^m}}(-1)^{f(x)+f(x+w)}=-2^m$. This implies that $w'=w$ and thus
we finish the proof of the second part of this lemma.
\end{proof}

The following lemma can be easily checked, so we omit the proof.
\begin{lemma}
Let $f_\lambda=Tr_1^m(\lambda x^d)$ be a balanced quadratic semi-bent function and $T$ be the set $\{\alpha\in\F_2^m:\widehat{f}(\alpha)=0\}$.
Them we have  $T^\perp=\{0,\lambda^{-\frac{1}{d}}\}$.
\end{lemma}

\begin{lemma}\label{L:distribution-AB-sum-subspace}
Let $f$ be a quadratic semi-bent function defined on $\F_{2^m}$, where $m\geq 5$.
Let $T$ be the set $\{\alpha\in\F_2^m:\widehat{f}(\alpha)=0\}$ and
$E\not=T$ be a subspace of $\F_2^m$ with dimension $m-1$ and $E^\perp=\{0,w\}$.
Define $N_{f}(E,t)=|\{a\in E: \widehat{f}(a)=t\}|$.
Then we have $N_{f}(E,0)=2^{m-2}$,
$N_{f}(E,2^{(m+1)/2})=2^{m-3}+2^{(m-3)/2}(1-f(0)-f(w))$ and
$N_{f}(E,-2^{(m+1)/2})=2^{m-3}-2^{(m-3)/2}(1-f(0)-f(w))$.
\end{lemma}
\begin{proof}
Note that
\begin{eqnarray*}
\sum_{\alpha\in E}\widehat{f}^2(\alpha)&=&\sum_{\alpha\in E}\sum_{x,y\in\F_{2^m}}(-1)^{f(x)+f(y)+Tr_1^m(\alpha(x+y))}\\
&=&2^{m-1}\sum_{\beta\in E^\perp}\sum_{x\in\F_{2^m}}(-1)^{f(x)+f(x+\beta)} \hbox{~~~(by Lemma \ref{L:linearfunction-sum-E})}\\
&=&2^{2m-1}+2^{m-1}\sum_{x\in\F_{2^m}}(-1)^{f(x)+f(x+w)}.
\end{eqnarray*}
It is easy to see that $w\not\in T^\perp$, thus we have $\sum_{x\in\F_{2^m}}(-1)^{f(x)+f(x+w)}=0$
by Lemma \ref{L:walsh-space-quadraticsemibent}.
So we have $\sum_{\alpha\in E}\widehat{f}^2(\alpha)=2^{2m-1}$. Recall that $f$ is semi-bent.
Thus, we have $\big(2^{(m+1)/2}\big)^2\big(N_{f}(E,2^{(m+1)/2})+N_{f}(E,-2^{(m+1)/2})\big)=2^{2m-1}$.
This implies that
\begin{eqnarray}\label{L:distribution-AB-sum-subspace-eq1}
N_{f}(E,2^{(m+1)/2})+N_{f}(E,-2^{(m+1)/2})=2^{m-2}.
\end{eqnarray}
Then we have $N_{f}(E,0)=2^{m-2}$. Note also that
\begin{eqnarray*}
\sum_{\alpha\in E}\widehat{f}(\alpha)=\sum_{\alpha\in E}\sum_{x\in\F_{2^m}}(-1)^{f(x)+Tr_1^m(\alpha x)}
=\sum_{x\in\F_{2^m}}(-1)^{f(x)}\sum_{\alpha\in E}(-1)^{Tr_1^m(\alpha x)}.
\end{eqnarray*}
Then by Lemma \ref{L:linearfunction-sum-E}, we have $\sum_{\alpha\in E}\widehat{f}(\alpha)=2^m(1-f(0)-f(w))$,
which is equivalent to saying that
\begin{eqnarray}\label{L:distribution-AB-sum-subspace-eq2}
2^{(m+1)/2}\big(N_{f}(E,2^{(m+1)/2})-N_{f}(E,-2^{(m+1)/2})\big)=2^m(1-f(0)-f(w)).
\end{eqnarray}

Combining Equations \eqref{L:distribution-AB-sum-subspace-eq1} and \eqref{L:distribution-AB-sum-subspace-eq2},
we can deduce that $N_{f}(E,2^{(m+1)/2})=2^{m-3}+2^{(m-3)/2}(1-f(0)-f(w))$ and
$N_{f}(E,-2^{(m+1)/2})=2^{m-3}-2^{(m-3)/2}(1-f(0)-f(w))$. This completes the proof.
\end{proof}

\begin{lemma}\label{L:squaresum-walsh} 
Let $f_\nu$, $f_\lambda$ and $f_\mu$ be three Boolean functions of
$m$ variables such that $f_\nu=f_\lambda+f_\mu$. Then we have
$\sum_{\alpha\in\F_{2^m}}\big(\widehat{f_\lambda}(\alpha)-\widehat{f_\mu}(\alpha)\big)^2=2^{m+2}n_{f_\nu}$,
$\sum_{\alpha\in\F_{2^m}}\big(\widehat{f_\lambda}^2(\alpha)+\widehat{f_\mu}^2(\alpha)\big)=2^{2m+2}-2^{m+2}n_{f_\nu}$,
where $n_{f_\nu}$ denotes the size of support of $f_\nu$.
\end{lemma}
\begin{proof}
We have
$
\sum_{\alpha\in\F_{2^m}}\big(\widehat{f_\lambda}(\alpha)-\widehat{f_\mu}(\alpha)\big)^2
=\sum_{\alpha\in\F_{2^m}}\big(\widehat{f_\lambda}^2(\alpha)+\widehat{f_\mu}^2(\alpha)-2\widehat{f_\lambda}(\alpha)\widehat{f_\mu}(\alpha)\big)
=\sum_{\alpha\in\F_{2^m}}\widehat{f_\lambda}^2(\alpha)+\sum_{\alpha\in\F_{2^m}}\widehat{f_\mu}^2(\alpha)-2\sum_{\alpha\in\F_{2^m}}\widehat{f_\lambda}(\alpha)\widehat{f_\mu}(\alpha).
$
Note that
$\sum_{\alpha\in\F_{2^m}}\widehat{f_\lambda}(\alpha)\widehat{f_\mu}(\alpha)
=\sum_{\alpha\in\F_{2^m}}\big(\sum_{x\in\F_{2^m}}(-1)^{f_\lambda(x)+Tr_1^m(\alpha x)}\sum_{y\in\F_{2^m}}(-1)^{f_\mu(y)+Tr_1^m(\alpha y)}\big)
=\sum_{x,y\in\F_{2^m}}(-1)^{f_\lambda(x)+f_\mu(y)}\\ \times\sum_{\alpha\in\F_{2^m}}(-1)^{Tr_1^m(\alpha (x+y))}
=2^m\sum_{x\in\F_{2^m}}(-1)^{f_\lambda(x)+f_\mu(x)}
=2^m\sum_{x\in\F_{2^m}}(-1)^{f_\nu(x)}
=2^m\widehat{f_\nu}(0)
=2^{2m}-2^{m+1}n_{f_\nu}.
$
By Parseval's relation \cite{MS1977} we have
$\sum_{\alpha\in\F_{2^m}}\widehat{f_\lambda}^2(\alpha)=\sum_{\alpha\in\F_{2^m}}\widehat{f_\mu}^2(\alpha)=2^{2m}$.
Thus, we have
$\sum_{\alpha\in\F_{2^m}}\big(\widehat{f_\lambda}(\alpha)-\widehat{f_\mu}(\alpha)\big)^2=
2^{m+2}n_{f_\nu}$ and
$\sum_{\alpha\in\F_{2^m}}\big(\widehat{f_\lambda}^2(\alpha)+\widehat{f_\mu}^2(\alpha)\big)=2^{2m+2}-2^{m+2}n_{f_\nu}$.
This completes the proof.\end{proof}

\begin{lemma}\label{L:distribution-AB-set}
Let $F(x)=x^e$ be the Gold functions in odd dimension $m\geq 5$, where $e=2^i+1$ and
$i$ is such that $\gcd(i,n)=1$ and $1\leq i\leq (n-1)/2$.
Let $f_\lambda=Tr_1^n(\lambda F)$ and $f_\mu=Tr_1^n(\mu F)$ be two
Boolean functions, where $\lambda,\mu$ are two distinct elements of $\F_{2^m}^*$. We define a
subset $\mathcal{S}_\lambda$ formed by the codewords as follows:
$$\mathcal{S}_\lambda=\big\{s_x: x\in \F_{2^m}\big\},$$
where  $s_x=\big(Tr_1^m(xd_1+\lambda F(d_1)),\cdots, Tr_1^m(xd_{2^{m-1}}+\lambda F(d_{2^{m-1}})\big)$
and $D_{f_{\lambda+\mu}^1}=\{d\in\F_{2^m}: Tr_1^m((\lambda+\mu)F(d))=1\}$.
Then the weight distribution of the codewords included in $\mathcal{S}_\lambda$
is shown as in Table \ref{Table:distribution-AB-set}.
   \begin{table}[h!]
    \begin{center}
    \caption{The weight distribution of the codewords included in $\mathcal{S}_\lambda$}\label{Table:distribution-AB-set}
    \begin{tabular}{|c|c|c|c|c|c|c|c|c|c|c|c|c|c|c|c}
    \hline  Weight $w$ & Multiplicity $A_w$
    \\\hline   $2^{m-2}-2^{(m-1)/2}$   &$2^{m-4}+2^{(m-5)/2}(Tr_1^m(\frac{\lambda}{\mu}) -Tr_1^m(\frac{\mu}{\lambda}))$
    \\\hline   $2^{m-2}-2^{(m-3)/2}$   &$2^{m-2}+2^{(m-3)/2}(Tr_1^m(\frac{\mu}{\lambda})-Tr_1^m(\frac{\lambda}{\mu}) )$
    \\\hline   $2^{m-2}$               &$3\cdot 2^{m-3}$
    \\\hline   $2^{m-2}+2^{(m-3)/2}$   &$2^{m-2}+2^{(m-3)/2}(Tr_1^m(\frac{\lambda}{\mu}) -Tr_1^m(\frac{\mu}{\lambda}))$
    \\\hline   $2^{m-2}+2^{(m-1)/2}$   &$2^{m-4}+2^{(m-5)/2}(Tr_1^m(\frac{\mu}{\lambda})-Tr_1^m(\frac{\lambda}{\mu}) )$
    \\ \hline
    \end{tabular}
    \end{center}
    \end{table}
\end{lemma}
\begin{proof}
Note that for any $x\in\F_{2^m}$ we have
\begin{eqnarray*}
\left\{\begin{array}{lll}
\sum\limits_{d\in\F_{2^m}\setminus D_{f_{\lambda+\mu}^1}}(-1)^{Tr_1^m(xd+\lambda F(d))}+\sum\limits_{d\in D_{f_{\lambda+\mu}^1}}(-1)^{Tr_1^m(xd+\lambda F(d))}&=&\widehat{f_\lambda}(x)\\
\sum\limits_{d\in\F_{2^m}\setminus D_{f_{\lambda+\mu}^1}}(-1)^{Tr_1^m(xd+\lambda F(d))}-\sum\limits_{d\in D_{f_{\lambda+\mu}^1}}(-1)^{Tr_1^m(xd+\lambda F(d))}&=&\widehat{f_\mu}(x)
\end{array} \right..
\end{eqnarray*}
This implies that, for any $x\in\F_{2^m}^*$, we have
\begin{eqnarray}\label{Eq:distribution-AB-set-wt}
\mathrm {wt}(s_x)&=&2^{m-2}-\frac{1}{4}\big(\widehat{f_\lambda}(x)-\widehat{f_\mu}(x)\big).
\end{eqnarray}
We can easily see that for any $x\in\F_{2^m}$, $\widehat{f_\lambda}(x)-\widehat{f_\mu}(x)\in\{0,\pm 2^{(m+1)/2},\pm 2^{(m+3)/2}\}$.
Therefore, for obtaining the weight distribution of the codewords in $\mathcal{S}_\lambda$, we need to
calculate the distribution of the values $0,\pm 2^{(m+1)/2},\pm 2^{(m+3)/2}$ in the set $\{\widehat{f_\lambda}(x)-\widehat{f_\mu}(x):x\in\F_{2^m}\}$.
For doing this, we denote
$T_0=\{\alpha\in\F_2^m:\widehat{f_\lambda}(\alpha)=0\}$, $T_0^\perp=\{0,\lambda^{-\frac{1}{e}}\}$,
$T_1=\F_2^m\setminus\{T_0\}$, $S_0=\{\alpha\in\F_2^m:\widehat{f_\mu}(\alpha)=0\}$,
$S_0^\perp=\{0,\mu^{-\frac{1}{e}}\}$, $S_1=\F_2^m\setminus\{T_0\}$.
We define
\begin{eqnarray*}
\footnotesize{
\begin{array}{lll}
c_1=\{\alpha\in T_0: (\widehat{f_\lambda}(\alpha),\widehat{f_\mu}(\alpha))=(0,0)\},&c_2=\{\alpha\in T_0: (\widehat{f_\lambda}(\alpha),\widehat{f_\mu}(\alpha))=(0,2^{\frac{m+1}{2}})\},\\
c_3=\{\alpha\in T_0: (\widehat{f_\lambda}(\alpha),\widehat{f_\mu}(\alpha))=(0,-2^{\frac{m+1}{2}})\},&c_4=\{\alpha\in T_1: (\widehat{f_\lambda}(\alpha),\widehat{f_\mu}(\alpha))=(2^{\frac{m+1}{2}},0)\},\\
c_5=\{\alpha\in T_1: (\widehat{f_\lambda}(\alpha),\widehat{f_\mu}(\alpha))=(2^{\frac{m+1}{2}},2^{\frac{m+1}{2}})\},&c_6=\{\alpha\in T_1: (\widehat{f_\lambda}(\alpha),\widehat{f_\mu}(\alpha))=(2^{\frac{m+1}{2}},-2^{\frac{m+1}{2}})\},\\
c_7=\{\alpha\in T_1: (\widehat{f_\lambda}(\alpha),\widehat{f_\mu}(\alpha))=(-2^{\frac{m+1}{2}},0)\},&c_8=\{\alpha\in T_1: (\widehat{f_\lambda}(\alpha),\widehat{f_\mu}(\alpha))=(-2^{\frac{m+1}{2}},2^{\frac{m+1}{2}})\},\\
c_9=\{\alpha\in T_1: (\widehat{f_\lambda}(\alpha),\widehat{f_\mu}(\alpha))=(-2^{\frac{m+1}{2}},-2^{\frac{m+1}{2}})\}.
\end{array}}
\end{eqnarray*}
If the values of $c_i$ for $1\leq i \leq 9$ are known, then we can obtain the distribution of
the values $0,\pm 2^{(m+1)/2},\pm 2^{(m+3)/2}$ in the set $\{\widehat{f_\lambda}(x)-\widehat{f_\mu}(x):x\in\F_{2^m}\}$
and hence give the weight distribution of the codewords in $\mathcal{S}_\lambda$.
We now compute the values of $c_i$ for $1\leq i \leq 9$.
By Lemma \ref{L:walsh-space-quadraticsemibent}, $T_0$ and $S_0$ are two subspaces of $\F_2^m$ with dimension $m-1$.
Note that $T_0\not=S_0$ since $\lambda^{-\frac{1}{3}}\not=\mu^{-\frac{1}{3}}$.
Then by Lemma \ref{L:distribution-AB-sum-subspace}, we have
\begin{eqnarray}\label{L:c1}
c_1&=& 2^{m-2},
\end{eqnarray}
\begin{eqnarray}\label{L:c2}
c_2 &=& 2^{m-3}+2^{(m-3)/2}(1-Tr_1^m(\frac{\mu}{\lambda})),
\end{eqnarray}
\begin{eqnarray}\label{L:c3}
c_3 &=& 2^{m-3}-2^{(m-3)/2}(1-Tr_1^m(\frac{\mu}{\lambda})),
\end{eqnarray}
\begin{eqnarray}\label{L:c4}
c_4&=&2^{m-3}+2^{(m-3)/2}(1-Tr_1^m(\frac{\lambda}{\mu}) ),
\end{eqnarray}
\begin{eqnarray}\label{L:c7}
c_7&=&2^{m-3}-2^{(m-3)/2}(1-Tr_1^m(\frac{\lambda}{\mu}) ).
\end{eqnarray}
By Lemma \ref{L:squaresum-walsh}, we have
$\sum_{\alpha\in\F_{2^m}}\big(\widehat{f_\lambda}(\alpha)-\widehat{f_\mu}(\alpha)\big)^2=2^{2m+1}$. This implies that
$2^{m+1}\big(c_2+c_3+c_4+c_7)\big)+2^{m+3}\big(c_6+c_8\big)=2^{2m+1}$.
Combining \eqref{L:c2}\eqref{L:c3}\eqref{L:c4}\eqref{L:c7}, we have
\begin{eqnarray}\label{L:c6addc8}
c_6+c_8&=&2^{m-3}.
\end{eqnarray}
Similarly, it follows from $\sum_{\alpha\in\F_{2^m}}\big(\widehat{f_\lambda}(\alpha)+\widehat{f_\mu}(\alpha)\big)^2=2^{2m+1}$ that
\begin{eqnarray}\label{L:c5addc9}
c_5+c_9 &=&2^{m-3}.
\end{eqnarray}
According to Lemma \ref{L:Anne}, we have
\begin{eqnarray}\label{L:c456}
c_4+c_5+c_6=c_2+c_5+c_8=2^{m-2}+2^{(m-3)/2}.
\end{eqnarray}
Combining \eqref{L:c456}\eqref{L:c2}\eqref{L:c4}, we have
\begin{eqnarray}\label{L:c6minusc8}
\left\{\begin{array}{lll}
c_5+c_8&=&2^{m-3}+2^{(m-3)/2}Tr_1^m(\frac{\mu}{\lambda})\\
c_5+c_6&=&2^{m-3}+2^{(m-3)/2}Tr_1^m(\frac{\lambda}{\mu})
\end{array} \right..
\end{eqnarray}
By \eqref{L:c6addc8} and \eqref{L:c6minusc8} we have
\begin{eqnarray}\label{L:c6}
c_6=2^{m-4}+2^{(m-5)/2}(Tr_1^m(\frac{\lambda}{\mu}) -Tr_1^m(\frac{\mu}{\lambda}))
\end{eqnarray}
and
\begin{eqnarray}\label{L:c8}
c_8=2^{m-4}+2^{(m-5)/2}(Tr_1^m(\frac{\mu}{\lambda})-Tr_1^m(\frac{\lambda}{\mu}) ).
\end{eqnarray}
Further, we have
\begin{eqnarray}\label{L:c5}
c_5=2^{m-4}+2^{(m-5)/2}{\Big (}Tr_1^m{\Big (}\frac{\lambda}{\mu}{\Big )} +Tr_1^m{\Big (}\frac{\mu}{\lambda}{\Big )}{\Big )}
\end{eqnarray}
according to \eqref{L:c456}\eqref{L:c2}\eqref{L:c8},
and
\begin{eqnarray}\label{L:c9}
c_9=2^{m-4}-2^{(m-5)/2}{\Big (}Tr_1^m{\Big (}\frac{\lambda}{\mu}{\Big )} +Tr_1^m{\Big (}\frac{\mu}{\lambda}{\Big )}{\Big )}
\end{eqnarray}
by \eqref{L:c9} and \eqref{L:c5addc9}.
Thus, we have obtained the values of $c_i$ for all $1\leq i\leq 9$ and then
we get the weight distribution of the codewords in $\mathcal{S}_\lambda$
according to the definitions of $c_i$ and \eqref{Eq:distribution-AB-set-wt}.
This completes the proof.
\end{proof}

For any integer $m>0$, the Kloosterman sums over $\F_{2^m}$ are defined as
$\mathcal{K}(a)=\sum_{x\in\F_{2^m}}(-1)^{Tr_1^m(x^{2^m-2}+\alpha x)}$, where $\alpha\in\F_{2^m}$.
In fact, the Kloosterman sums are generally defined on the multiplicative
group $\F_{2^m}^*$. We extend them to $0$ by assuming $(-1)^0=1$.
\begin{lemma}[\cite{CarlitzKloo1969}]\label{L:Kloostermansumsone}
For any integer $m>0$, $\mathcal{K}(1)=1-\sum_{t=0}^{\lfloor m/2\rfloor}(-1)^{m-t}\frac{m}{m-t}{{m-t}\choose {t}}2^t$.
\end{lemma}

By the definition of Kloosterman sums, the following lemma can be easily obtained.

\begin{lemma}\label{L:Kloostermansumssubspace}
For any integer $m>0$, we have $\sum_{x\in\{z: Tr_1^m(z)=0\}}(-1)^{Tr_1^m(1/x)}=\frac{1}{2}\mathcal{K}(1)$
and $\sum_{x\in\{z: Tr_1^m(z)=1\}}(-1)^{Tr_1^m(1/x)}=-\frac{1}{2}\mathcal{K}(1)$.
\end{lemma}

\begin{lemma}\label{L:kluoospairs}
For any odd integer $m>0$ and arbitrary $\mu\in\F_{2^m}^*$, we
denote by $H'_\mu$ the set $\{x\in\F_{2^m}:Tr_1^m(\mu^{-1}x)=0\}$
and by $A_{(i,j)}$ the number of $(i,j)\in\Z_2\times\Z_2$ appeared
in the multi-set $\{(Tr_1^m(x/(x+\mu)), Tr_1^m((x+\mu)/x))\in\Z_2\times\Z_2:x\in H'_\mu\}$.
Then the values of $A_{(i,j)}$ is given in Table \ref{Table:kluoospairs}.
   \begin{table}[h!]
    \begin{center}
    \caption{The distribution of the element $(i,j)\in P_\mu$}\label{Table:kluoospairs}
    \begin{tabular}{|c|c|c|c|c|c|c|c|c|c|c|c|c|c|c|c}
    \hline  $(i,j)$ & $A_{(i,j)}$
    \\\hline   $(0,0)$ & $2^{m-3}+\frac{1}{8}\mathcal{K}(1)+\frac{1}{2}$
    \\\hline   $(0,1)$ & $2^{m-3}+\frac{1}{8}\mathcal{K}(1)-\frac{1}{2}$
    \\\hline   $(1,0)$ & $2^{m-3}-\frac{3}{8}\mathcal{K}(1)+\frac{1}{2}$
    \\\hline   $(1,1)$ & $2^{m-3}+\frac{1}{8}\mathcal{K}(1)-\frac{1}{2}$
     \\ \hline
    \end{tabular}
    \end{center}
    \end{table}
\end{lemma}
\begin{proof}
Note that $H'_\mu=\{\mu(y+y^2):y\in\F_{2^m}\}$.
We have
\begin{eqnarray*}
\sum_{x\in H'_\mu}(-1)^{Tr_1^m(\frac{x}{x+\mu})}&=&\frac{1}{2}\sum_{y\in \F_{2^m}}(-1)^{Tr_1^m(\frac{\mu(y+y^2)}{\mu(y+y^2)+\mu})}\\
&=&\frac{1}{2}\sum_{y\in \F_{2^m}}(-1)^{Tr_1^m(\frac{y+y^2}{y+y^2+1})}\\
&=&\sum_{x\in H'_1}(-1)^{Tr_1^m(\frac{x}{x+1})}\\
&=&\sum_{x\in\{z: Tr_1^m(z)=1\}}(-1)^{Tr_1^m(\frac{x+1}{x})}~~\mbox{(by changing $x$ into $x+1$)}\\
&=&-\sum_{x\in\{z: Tr_1^m(z)=1\}}(-1)^{Tr_1^m(\frac{1}{x})}\\
&=&\frac{1}{2}\mathcal{K}(1).
\end{eqnarray*}
The last identity follows from Lemma \ref{L:Kloostermansumssubspace}.
Similarly, we have
$$\sum_{x\in H'_\mu}(-1)^{Tr_1^m(\frac{x+\mu}{x})}=-\frac{1}{2}\mathcal{K}(1)+2~\mbox{and}~
\sum_{x\in H'_\mu}(-1)^{Tr_1^m\big(\frac{x}{x+\mu}+\frac{x+\mu}{x}\big)}=\frac{1}{2}\mathcal{K}(1).$$
Then we can easily obtain Table \ref{Table:kluoospairs} according to the values of
$\sum_{x\in H'_\mu}(-1)^{Tr_1^m(\frac{x}{x+\mu})}$,
$\sum_{x\in H'_\mu}(-1)^{Tr_1^m(\frac{x+\mu}{x})}$ and $\sum_{x\in H'_\mu}(-1)^{Tr_1^m\big(\frac{x}{x+\mu}+\frac{x+\mu}{x}\big)}$.

\end{proof}

\begin{theorem}\label{T:wdgold}
Let $F(x)=x^e$ be the Gold functions in odd dimension $m\geq 5$, where $e=2^i+1$ and
$i$ is such that $\gcd(i,n)=1$ and $1\leq i\leq (n-1)/2$.
Let $f_\nu=Tr_1^n(\nu F)$ and $H_\nu=\{x\in\F_{2^m}:Tr_1^m(\nu^{-1}x)=0\}$,
where $\nu\in\F_{2^m}^*$. Then $\mathcal{C}_{D_{f_{\nu}},H_\nu}$ given by
\eqref{E:codeVBF} is a $[2^{m-1}, 2m-1, 2^{m-2}-2^{(m-1)/2}]$ five-weight binary code
with the weight distribution in Table \ref{Table:wdgold}.
\end{theorem}
\begin{table}[h!]
\begin{center}
\caption{The weight distribution of the code of Theorem \ref{T:wdgold}}\label{Table:wdgold}
\begin{tabular}{|c|c|c|c|c|c|c|c|c|c|c|c|c|c|c|c}
    \hline  Weight $w$ & Multiplicity $A_w$
    \\\hline   $0$   &$1$
    \\\hline   $2^{m-2}-2^{(m-1)/2}$   &$2^{2m-5}+2^{\frac{m-5}{2}}-2^{\frac{m-7}{2}}\mathcal{K}(1)-2^{m-4}$
    \\\hline   $2^{m-2}-2^{(m-3)/2}$   &$2^{2m-3}+2^{\frac{m-5}{2}}\mathcal{K}(1)-2^{\frac{m-1}{2}}$
    \\\hline   $2^{m-2}$               &$3\cdot 2^{2m-4}+2^{m-3}-1$
    \\\hline   $2^{m-2}+2^{(m-3)/2}$   &$2^{2m-3}-2^{\frac{m-5}{2}}\mathcal{K}(1)+2^{\frac{m-1}{2}}$
    \\\hline   $2^{m-2}+2^{(m-1)/2}$   &$2^{2m-5}-2^{\frac{m-5}{2}}+2^{\frac{m-7}{2}}\mathcal{K}(1)-2^{m-4}$
    \\ \hline
\end{tabular}
\\ $*$\small{where the value of $\mathcal{K}(1)$ is given in Lemma \ref{L:Kloostermansumsone}.~~~~~~~~~~~~}
\end{center}
\end{table}

\begin{proof}
Recall that $F$ is a permutation over $\F_{2^m}$ and hence $f_\nu$ is balanced.
This implies that the length of code  $\mathcal{C}_{D_{f_{\nu}},H_\nu}$ is equal to $2^{m-1}$.
Recall also that $F$ has nonlinearity $2^{m-1}-2^{m/2-1}$.
By Theorem \ref{T:codeVBF} we have $\mathcal{C}_{D_{f_\nu},H_\nu}$ is
a $[2^{m-1}, 2m-1, 2^{m-2}-2^{(m-1)/2}]$ five-weight binary code.

We now determine the weight distribution of $\mathcal{C}_{D_{f_\nu},H_\nu}$.
We can see that the code $\mathcal{C}_{D_{f_\nu},H_\nu}$
does not include the all-one codeword.
Thus, according to the proof of Theorem \ref{T:wdPN0all},
we can assume that the nonzero codewords in $\mathcal{C}_{D_{f_\nu},H_\nu}$ have weights:
$$w_1=2^{m-2}-2^{\frac{m-1}{2}},w_2=2^{m-2}-2^{\frac{m-3}{2}},
w_3=2^{m-2},w_4=2^{m-2}+2^{\frac{m-3}{2}},w_5=2^{m-2}+2^{\frac{m-1}{2}}.$$
We now see the number $A_{w_i}$ of codewords with weight $w_i$
in $\mathcal{C}_{D_{f_{\nu}},H_\nu}$, where $i=1,2,3,4,5$.
Note that $\mathcal{C}_{D_{f_{\nu}},H_\nu}=\big(\cup_{\lambda\in H_\nu\setminus\{0\}}\mathcal{S}_\lambda\big)\cup \mathcal{C}'$,
where $\mathcal{S}_\lambda$, $\lambda\in H_\nu\setminus\{0\}$, is defined in Lemma \ref{L:distribution-AB-set} by replacing
$\lambda+\mu$ by $\nu$ and $\mathcal{C}'$ is defined as $\mathcal{C}'=\{c_{x,0}\in\mathcal{C}_{D_{f_\nu},H_\nu}: x\in\F_{2^m}\}$.
We can see that $\mathcal{C}'$ is a subcode of $\mathcal{C}_{D_{f_\lambda},H_\lambda}$
with length $2^{m-1}$ and dimension $m$.
By \eqref{E:pairsdHsetlambda0} and \eqref{E:pairsdHsetnotlambda0}
we have the nonzero codewords in $\mathcal{C}'$ only have weights $w_2,w_3,w_4$.
We now determine the number $A_{w_i}'$ of codewords with weights $w_2,w_3,w_4$.
It can be easily seen from \cite[Theorem 10]{MS1977} that the minimum weight of the dual code of $\mathcal{C}'$ is no less than $3$.
According to the first three Pless Power Moments \cite[p. 260]{HPbook2003}, we have
\begin{eqnarray*}
\left\{\begin{array}{lll}
A_{w_2}'+A_{w_3}'+A_{w_4}'&=&2^{m}-1\\
w_2A_{w_2}'+w_3A_{w_3}'+w_4A_{w_4}'&=&2^{2m-2}\\
w_2^2A_{w_2}'+w_3^2A_{w_3}'+w_4^2A_{w_4}'&=&(2^{m-1}+1)2^{2m-3}
\end{array} \right..
\end{eqnarray*}
By solving this system of equations, we have
$A_{w_2}'=2^{m-2}-2^{(m-3)/2}$, $A_{w_3}'=2^{m-1}-1$ and $A_{w_3}'=2^{m-2}+2^{(m-3)/2}$.
By \eqref{E:pairsdHsetlambda0} and \eqref{E:pairsdHsetnotlambda0}, we can see that the
codewords with weight $w_1,w_5$ only appear in the set $\cup_{\lambda\in H_\nu\setminus\{0\}}\mathcal{S}_\lambda$.
By Lemma \ref{L:distribution-AB-set}, we have
\begin{eqnarray*}
A_{w_1}&=&\sum_{z\in H_\nu\setminus\{0\}}\Big(2^{m-4}+2^{\frac{m-5}{2}} \big(Tr_1^m(\frac{z}{z+\nu}) -Tr_1^m(\frac{z+\nu}{z})\big)\Big)\\
&=&\Big(2^{m-4}\big(A_{(0,0)}-1\big)\Big)+\Big(\big(2^{m-4}-2^{\frac{m-5}{2}}\big)A_{(0,1)}\Big)\\
&&+\Big(\big(2^{m-4}+2^{\frac{m-5}{2}}\big)A_{(1,0)}\Big)+\Big(2^{m-4}A_{(1,1)}\Big)\\
&=&2^{m-4}|H_{\nu}-1|+2^{\frac{m-5}{2}}\big(1-\frac{1}{2}\mathcal{K}(1)\big)\\
&=&2^{2m-5}+2^{\frac{m-5}{2}}\big(1-\frac{1}{2}\mathcal{K}(1)\big)-2^{m-4}\\
&=&2^{2m-5}+2^{\frac{m-5}{2}}-2^{\frac{m-7}{2}}\mathcal{K}(1)-2^{m-4}
\end{eqnarray*}
in which the values of $A_{(i,j)}$ come from Table \ref{Table:kluoospairs}.
Similarly, we could deduce that
\begin{eqnarray*}
A_{w_5}&=&2^{2m-5}+2^{\frac{m-7}{2}}\mathcal{K}(1)-2^{\frac{m-5}{2}}-2^{m-4}.
\end{eqnarray*}
We now calculate the values of $A_{w_2},A_{w_3},A_{w_4}$. Note that
the codewords with weight $w_2,w_3,w_4$ appear in both
$\cup_{\lambda\in H_\nu\setminus\{0\}}\mathcal{S}_\lambda$ and $\C_{D_{f_\nu}}$.
By Lemma \ref{L:distribution-AB-set} and recall that $A_{w_2}'=2^{m-2}-2^{(m-3)/2}$, we have
\begin{eqnarray*}
A_{w_2}&=&\sum_{z\in H_\nu\setminus\{0\}}\Big(2^{m-2}+2^{\frac{m-3}{2}} \big(Tr_1^m(\frac{z+\nu}{z})-Tr_1^m(\frac{z}{z+\nu})\big)\Big)+\Big(2^{m-2}-2^{\frac{m-3}{2}}\Big)\\
&=&\Big(2^{m-2}\big(A_{(0,0)}-1\big)+\big(2^{m-2}+2^{\frac{m-3}{2}}\big)A_{(0,1)}+\big(2^{m-2}-2^{\frac{m-3}{2}}\big)A_{(1,0)}+2^{m-2}A_{(1,1)}\Big)\\
&&+\Big(2^{m-2}-2^{\frac{m-3}{2}}\Big)\\
&=&\Big(2^{m-2}|H_{\nu}-1|+2^{\frac{m-3}{2}}\big(\frac{1}{2}\mathcal{K}(1)-1\big)\Big)+\Big(2^{m-2}-2^{\frac{m-3}{2}}\Big)\\
&=&2^{2m-3}+2^{\frac{m-5}{2}}\mathcal{K}(1)-2^{\frac{m-1}{2}},
\end{eqnarray*}
in which the values of $A_{(i,j)}$ are given in Table \ref{Table:kluoospairs}.
Similarly, we have
\begin{eqnarray*}
A_{w_4}&=&2^{2m-3}-2^{\frac{m-5}{2}}\mathcal{K}(1)+2^{\frac{m-1}{2}}
\end{eqnarray*}
and
\begin{eqnarray*}
A_{w_3}&=&\Big(3\cdot2^{m-3}|H_{\nu}-1|\Big)+\Big(2^{m-1}-1\Big)\\
&=&3\cdot 2^{2m-4}+2^{m-3}-1.
\end{eqnarray*}
This completes the proof.
\end{proof}

\begin{example}\label{E:ABsubcode}
For $m=9$, we define $F(x)=x^3$, where $x\in\F_{2^9}$.
Let $\lambda=1$ in \eqref{E:codeVBF}. Thus we have $f_{\lambda}=Tr_1^9(x^3)$
and $H_{\lambda}=\{x\in\F_{2^9} : Tr_1^9(x)=0\}$.
By our Magma program, we can get the linear code $\mathcal{C}_{D_{f_\lambda},H_\lambda}$ defined by
\eqref{E:codeVBF} is a $[256,17,112]$-code with weight enumerator
$1+8172z^{112}+32736z^{120}+49215z^{128}+32800z^{136}+8148z^{144}$,
which confirms the result of Theorem \ref{T:wdgold}.
\end{example}

\section{Conclusion}\label{sec:conc}

Inspired by a generic recent construction developed by Ding \emph{et al.}, we constructed several classes of binary linear codes
from vectorial Boolean functions and determined their parameters.
Firstly, by employing PN functions and AB functions we obtained several classes of six-weight linear codes which contain the
all-one codeword. Secondly, we defined a subcode in any linear code we constructed and considered its parameter.
When the vectorial Boolean function is a PN function or a Gold AB function (in odd dimension),
we completely determined the weight distribution of this subcode.
Besides, our linear codes have more larger dimensions than the ones by Ding \emph{et al.}'s generic construction.


\end{document}